\PassOptionsToPackage{unicode}{hyperref}
\PassOptionsToPackage{hyphens}{url}
\documentclass[
  english,
]{article}
\usepackage{xcolor}
\usepackage{amsmath,amssymb}
\usepackage{iftex}
\ifPDFTeX
  \usepackage[T1]{fontenc}
  \usepackage[utf8]{inputenc}
  \usepackage{textcomp} 
\else 
  \usepackage{unicode-math} 
  \defaultfontfeatures{Scale=MatchLowercase}
  \defaultfontfeatures[\rmfamily]{Ligatures=TeX,Scale=1}
\fi
\usepackage{lmodern}
\ifPDFTeX\else
\fi
\IfFileExists{upquote.sty}{\usepackage{upquote}}{}
\IfFileExists{microtype.sty}{
  \usepackage[]{microtype}
  \UseMicrotypeSet[protrusion]{basicmath} 
}{}
\makeatletter
\@ifundefined{KOMAClassName}{
  \IfFileExists{parskip.sty}{%
    \usepackage{parskip}
  }{
    \setlength{\parindent}{0pt}
    \setlength{\parskip}{6pt plus 2pt minus 1pt}}
}{
  \KOMAoptions{parskip=half}}
\makeatother
\makeatletter
\ifx\paragraph\undefined\else
  \let\oldparagraph\paragraph
  \renewcommand{\paragraph}{
    \@ifstar
      \xxxParagraphStar
      \xxxParagraphNoStar
  }
  \newcommand{\xxxParagraphStar}[1]{\oldparagraph*{#1}\mbox{}}
  \newcommand{\xxxParagraphNoStar}[1]{\oldparagraph{#1}\mbox{}}
\fi
\ifx\subparagraph\undefined\else
  \let\oldsubparagraph\subparagraph
  \renewcommand{\subparagraph}{
    \@ifstar
      \xxxSubParagraphStar
      \xxxSubParagraphNoStar
  }
  \newcommand{\xxxSubParagraphStar}[1]{\oldsubparagraph*{#1}\mbox{}}
  \newcommand{\xxxSubParagraphNoStar}[1]{\oldsubparagraph{#1}\mbox{}}
\fi
\makeatother

\usepackage{longtable,booktabs,array}
\usepackage{calc} 
\usepackage{etoolbox}
\makeatletter
\patchcmd\longtable{\par}{\if@noskipsec\mbox{}\fi\par}{}{}
\makeatother
\IfFileExists{footnotehyper.sty}{\usepackage{footnotehyper}}{\usepackage{footnote}}
\makesavenoteenv{longtable}
\usepackage{graphicx}
\makeatletter
\newsavebox\pandoc@box
\newcommand*\pandocbounded[1]{
  \sbox\pandoc@box{#1}%
  \Gscale@div\@tempa{\textheight}{\dimexpr\ht\pandoc@box+\dp\pandoc@box\relax}%
  \Gscale@div\@tempb{\linewidth}{\wd\pandoc@box}%
  \ifdim\@tempb\p@<\@tempa\p@\let\@tempa\@tempb\fi
  \ifdim\@tempa\p@<\p@\scalebox{\@tempa}{\usebox\pandoc@box}%
  \else\usebox{\pandoc@box}%
  \fi%
}
\def\fps@figure{htbp}
\makeatother

\ifLuaTeX
\usepackage[bidi=basic]{babel}
\else
\usepackage[bidi=default]{babel}
\fi
\ifPDFTeX
\else
\babelfont{rm}[]{Latin Modern Roman}
\fi

\def\languageshorthands#1{}
\ifLuaTeX
  \usepackage[english]{selnolig} 
\fi

\providecommand{\tightlist}{%
  \setlength{\itemsep}{0pt}\setlength{\parskip}{0pt}}

\usepackage[style=authoryear,backend=biber,backref=true]{biblatex}
\usepackage{arxiv}
\usepackage{orcidlink}
\usepackage{amsmath}
\usepackage[T1]{fontenc}
\makeatletter
\@ifpackageloaded{caption}{}{\usepackage{caption}}
\AtBeginDocument{%
\ifdefined\contentsname
  \renewcommand*\contentsname{Table of contents}
\else
  \newcommand\contentsname{Table of contents}
\fi
\ifdefined\listfigurename
  \renewcommand*\listfigurename{List of Figures}
\else
  \newcommand\listfigurename{List of Figures}
\fi
\ifdefined\listtablename
  \renewcommand*\listtablename{List of Tables}
\else
  \newcommand\listtablename{List of Tables}
\fi
\ifdefined\figurename
  \renewcommand*\figurename{Figure}
\else
  \newcommand\figurename{Figure}
\fi
\ifdefined\tablename
  \renewcommand*\tablename{Table}
\else
  \newcommand\tablename{Table}
\fi
}
\@ifpackageloaded{float}{}{\usepackage{float}}
\floatstyle{ruled}
\@ifundefined{c@chapter}{\newfloat{codelisting}{h}{lop}}{\newfloat{codelisting}{h}{lop}[chapter]}
\floatname{codelisting}{Listing}

\usepackage{amsthm}
\theoremstyle{definition}
\newtheorem{definition}{Definition}
\theoremstyle{plain}
\newtheorem{proposition}{Proposition}
\theoremstyle{definition}
\newtheorem{example}{Example}
\theoremstyle{plain}
\newtheorem{corollary}{Corollary}
\theoremstyle{remark}
\AtBeginDocument{}
\newtheorem*{remark}{Remark}

\makeatother
\makeatletter
\@ifpackageloaded{caption}{}{\usepackage{caption}}
\@ifpackageloaded{subcaption}{}{\usepackage{subcaption}}
\makeatother
\usepackage{bookmark}
\IfFileExists{xurl.sty}{\usepackage{xurl}}{} 
\hypersetup{
  pdftitle={ICS for complex data with application to outlier detection for density data},
  pdfauthor={Camille Mondon; Huong Thi Trinh; Anne Ruiz-Gazen; Christine Thomas-Agnan},
  pdflang={en},
  pdfkeywords={Bayes spaces, Distributional data, Extreme
weather, Functional data, Invariant coordinate selection, Outlier
detection, Temperature distribution, 62H25, 62R10, 62G07, 65D07},
  hidelinks,
  pdfcreator={LaTeX via pandoc}}

\renewcommand{\today}{May 20, 2025}
\newcommand{\runninghead}{A Preprint }
\title{ICS for complex data with application to outlier detection for
density data}
\def\asep{\\\\\\ } 
\def\asep{\And }
\author{\textbf{Camille
Mondon}~\orcidlink{0009-0007-4569-990X}\\Mathematics and
Statistics\\Toulouse School of
Economics\\Toulouse,\ 31000\\\href{mailto:camille.mondon@tse-fr.eu}{camille.mondon@tse-fr.eu}\asep\textbf{Huong
Thi Trinh}~\orcidlink{0000-0002-5615-5787}\\Faculty of Mathematical
Economics\\Thuongmai
University\\Hanoi\\\href{mailto:trinhthihuong@tmu.edu.vn}{trinhthihuong@tmu.edu.vn}\asep\textbf{Anne
Ruiz-Gazen}~\orcidlink{0000-0001-8970-8061}\\Mathematics and
Statistics\\Toulouse School of
Economics\\Toulouse,\ 31000\\\href{mailto:anne.ruiz-gazen@tse-fr.eu}{anne.ruiz-gazen@tse-fr.eu}\asep\textbf{Christine
Thomas-Agnan}~\orcidlink{0000-0002-6430-3110}\\Mathematics and
Statistics\\Toulouse School of
Economics\\Toulouse,\ 31000\\\href{mailto:christine.thomas@tse-fr.eu}{christine.thomas@tse-fr.eu}}
\date{May 20, 2025}
\begin{document}
\maketitle
\begin{abstract}
Invariant coordinate selection (ICS) is a dimension reduction method,
used as a preliminary step for clustering and outlier detection. It has
been primarily applied to multivariate data. This work introduces a
coordinate-free definition of ICS in an abstract Euclidean space and
extends the method to complex data. Functional and distributional data
are preprocessed into a finite-dimensional subspace. For example, in the
framework of Bayes Hilbert spaces, distributional data are smoothed into
compositional spline functions through the Maximum Penalised Likelihood
method. We describe an outlier detection procedure for complex data and
study the impact of some preprocessing parameters on the results. We
compare our approach with other outlier detection methods through
simulations, producing promising results in scenarios with a low
proportion of outliers. ICS allows detecting abnormal climate events in
a sample of daily maximum temperature distributions recorded across the
provinces of Northern Vietnam between 1987 and 2016.
\end{abstract}
{\bfseries \emph Keywords}
\def\sep{\textbullet\ }
Bayes spaces \sep Distributional data \sep Extreme
weather \sep Functional data \sep Invariant coordinate
selection \sep Outlier detection \sep Temperature
distribution \sep 62H25 \sep 62R10 \sep 62G07 \sep 
65D07

\section{Introduction}\label{introduction}

The invariant coordinate selection (ICS) method was introduced in a
multivariate data analysis framework by \textcite{tyler_invariant_2009}.
ICS is one of the dimension reduction methods that extend beyond
Principal Component Analysis (PCA) and second moments. ICS seeks
projection directions associated with the largest and/or smallest
eigenvalues of the simultaneous diagonalisation of two scatter matrices
\autocites[see][]{loperfido_theoretical_2021}[ for recent
references]{nordhausen_usage_2022}. This approach enables ICS to uncover
underlying structures, such as outliers and clusters, that might be
hidden in high-dimensional spaces. ICS is termed ``invariant'' because
it produces components, linear combinations of the original features of
the data, that remain invariant (up to their sign and some permutation)
under affine transformations of the data, including translations,
rotations and scaling. Moreover, Theorem 4 in
\autocite{tyler_invariant_2009} demonstrates that, for a mixture of
elliptical distributions, the projection directions of ICS associated
with the largest or smallest eigenvalues usually generate the Fisher
discriminant subspace, regardless of the chosen pair of scatter matrices
and without prior knowledge of group assignments. Once the pair of
scatter matrices is chosen, invariant components can be readily
computed, and dimension reduction is achieved by selecting the
components that reveal the underlying structure. Recent articles have
examined in detail the implementation of ICS in a multivariate
framework, focusing on objectives such as anomaly detection
\autocite{archimbaud_ics_2018} or clustering
\autocite{alfons_tandem_2024}. These studies particularly address the
choice of pairs of scatter matrices and the selection of relevant
invariant components. Note that this idea of joint diagonalisation of
scatter matrices is also used in the context of blind source separation
and more precisely for Independent Component Analysis (ICA) which is a
model-based approach as opposed to ICS \autocite[see][ for more
details]{nordhausen_usage_2022}. ICS has later been adapted to more
complex data, namely compositional data
\autocite{ruiz-gazen_detecting_2023}, functional data
\autocites{rendon_aguirre_clustering_2017}[ for ICA]{li_functional_2021}
and multivariate functional data \autocites{archimbaud_ics_2022}[ for
ICA]{virta_independent_2020}.

A significant contribution of the present work is the formulation of a
coordinate-free variant of ICS, considering data objects in an abstract
Euclidean space, without having to choose a specific basis. This
formulation allows ICS to be consistently defined in a very general
framework, unifying its original definition for multivariate data and
its past adaptations to specific types of complex data. In the case of
compositional data, the coordinate-free approach yields an alternative
implementation of ICS that is more computationally efficient. We are
also able to propose a new version of invariant coordinate selection
adapted to distributional data. Note that a coordinate-free version of
ICS has already been mentioned in \autocite{tyler_invariant_2009}, in
the discussion by Mervyn Stone, who proposed to follow the approach of
\textcite{stone_coordinate-free_1987}. In their response, Tyler and
co-authors agree that this could offer a theoretically elegant and
concise view of the topic. A coordinate-free approach of ICA is proposed
by \textcite{li_functional_2021}, but to our knowledge, no
coordinate-free approach to ICS exists for a general Euclidean space.

As mentioned above, a possible application of ICS is outlier detection.
In the context of a small proportion of outliers, a complete detection
procedure integrating a dimension reduction step based on the selection
of invariant coordinates is described by \textcite{archimbaud_ics_2018}.
This method, called ICSOutlier, flags outlying observations and has been
implemented for multivariate data by
\textcite{nordhausen_icsoutlier_2023}. It has been adapted to
compositional data by \textcite{ruiz-gazen_detecting_2023} and to
multivariate functional data by \textcite{archimbaud_ics_2022}. We
propose to extend this detection procedure to complex data and
illustrate it on distributional data.

Detecting outliers is already challenging in a classical multivariate
context because outliers may differ from the other observations in their
correlation pattern \autocite[see][ for an overview on outlier detection
and analysis]{aggarwal_outlier_2017}. \textcite{archimbaud_ics_2018}
demonstrate how the ICS procedure outperforms those based on the
Mahalanobis distance and PCA (robust or not). For compositional data,
the constraints of positivity and constant sum of coordinates must be
taken into account as detailed in \autocite{ruiz-gazen_detecting_2023}
and further examined in this paper. For univariate functional data,
outliers are categorised as either magnitude or shape outliers, with
shape outliers being more challenging to detect because they are hidden
among the other curves. Many existing detection methods for functional
data rely on depth measures, including the Mahalanobis distance
\autocite[see, e.g., the recent paper][ and the included
references]{dai_functional_2020}. Density data are constrained
functional data, and thus combine the challenges associated with both
compositional and functional data. The literature on outlier detection
for density data is very sparse and recent with, as far as we know, the
papers by \textcite{menafoglio_anomaly_2021},
\textcite{lei_functional_2023} and \textcite{murph_visualisation_2024}
only. Two types of outliers have been identified for density data: the
horizontal-shift outliers and the shape outliers, with shape outliers
being again more challenging to detect \autocite[see][ for
details]{lei_functional_2023}. The procedure proposed by
\textcite{menafoglio_anomaly_2021} is based on an adapted version of
functional PCA to density objects in a control chart context. In order
to derive a robust distribution-to-distribution regression method,
\textcite{lei_functional_2023} propose a transformation tree approach
that incorporates many different outlier detection methods adapted to
densities. Their methods involve transforming density data into
unconstrained data and using standard functional outlier detection
methods. \textcite{murph_visualisation_2024} continue the work of the
previously cited article by comparing more methods through simulations,
and give an application to gas transport data. ICS is not mentioned in
these references.

Our coordinate-free definition of ICS enables direct adaptation of the
ICSOutlier method to complex data. Through a case study on temperature
distributions in Vietnam, we assess the impact of preprocessing
parameters and provide practical recommendations for their selection. In
addition, the results of a simulation study demonstrate that our method
performs favourably compared with other approaches. An original
application to Vietnamese data provides a detailed description of the
various stages involved in detecting low-proportion outliers using ICS,
as well as interpreting them from the dual eigendensities.

Section~\ref{sec-icsfree} presents ICS in a coordinate-free framework,
states a useful result to link ICS in different spaces, and treats the
specific cases of compositional, functional and distributional data. For
the latter, we develop a Bayes space approach and discuss the maximum
penalised likelihood method to preprocess the original samples of
real-valued data into a sample of compositional splines.
Section~\ref{sec-icscomplexout} describes the ICS-based outlier
detection procedure adapted to complex data, discusses the impact of the
preprocessing parameters on outlier detection through a toy example.
Simulating data from multiple generating schemes, we compare ICS with
other outlier detection methods for density data.
Section~\ref{sec-appliout} provides an application of the outlier
detection methodology to maximum temperature data in Vietnam over 30
years. Section~\ref{sec-conclusion} concludes the paper and offers some
perspectives. Supplementary material on ICS, a reminder on Bayes spaces,
as well as proofs of the propositions and corollaries are given in the
Appendix.

\section{ICS for complex data}\label{sec-icsfree}

A naive approach to ICS for complex data would be to apply multivariate
ICS to coordinate vectors in a basis. This not only ignores the metric
on the space when the basis is not orthonormal, but also gives a
potentially different ICS method for each choice of basis \autocite[as
in][]{archimbaud_ics_2022}. Defining a unique coordinate-free ICS
problem avoids defining multiple ICS methods and having to discuss the
potential links between them, thus making our approach more intrinsic.
In particular, it leads to more interpretable invariant components that
are of the same nature as the considered complex random objects. In the
case of functional or distributional data, the usual framework assumes
that the data objects reside in an infinite-dimensional Hilbert space,
which leads to non-orthonormal bases and incomplete inner product
spaces. We choose to restrict our attention to finite-dimensional
approximations of the data in the framework of Euclidean spaces, which
are particularly suitable here because ICS is known to fail when the
dimension is larger than the sample size \autocite{tyler_note_2010}.
This suggests that an ICS method for infinite-dimensional Hilbert spaces
would require modifying the core of the method, which is beyond the
scope of this work.

\subsection{A coordinate-free ICS problem}\label{sec-icsproblem}

In order to generalise invariant coordinate selection \autocite[
def.~1]{tyler_invariant_2009} to a coordinate-free framework in a
Euclidean space \(E\), we need to eliminate any reference to a
coordinate system, which means replacing coordinate vectors by abstract
vectors, matrices by linear mappings, bases or quadratic forms,
depending on the context. This coordinate emancipation procedure will
ensure that our definition of ICS for an \(E\)-valued random object
\(X\) does not depend on any particular choice of basis of \(E\) to
represent \(X\).

Following this methodology, we are able to immediately generalise the
definition of (affine equivariant) scatter operators from random vectors
in \(E = \mathbb R^p\) \autocite[as defined in][
eq.~3]{tyler_invariant_2009} to random objects in a Hilbert space \(E\).
This is a perfect example of how the coordinate-free framework can be
used to extend existing work to infinite-dimensional spaces. For further
details, see Definition~\ref{def-scatter} in the Appendix. A notable
difference from \autocite{tyler_invariant_2009} is that we work directly
with random objects instead of their underlying distributions. In
particular, we introduce an affine invariant space \(\mathcal E\) of
random objects on which the scatter operators are defined and to which
we assume that \(X\) belongs. For example,
\(\mathcal E = L^p (\Omega, E)\) corresponds to assuming the existence
of the \(p\) first moments of \(\| X \|\).

Again, emancipating from coordinates allows us to naturally generalise
ICS to complex random objects in a Euclidean space.

\begin{definition}[Coordinate-free
ICS]\protect\hypertarget{def-ics}{}\label{def-ics}

Let \((E, \langle \cdot, \cdot \rangle)\) be a Euclidean space of
dimension \(p\), \(\mathcal E \subseteq L^1 (\Omega, E)\) an affine
invariant set of integrable \(E\)-valued random objects, \(S_1\) and
\(S_2\) two scatter operators on \(\mathcal E\) and
\(X \in \mathcal E\). The invariant coordinate selection problem
\(\operatorname{ICS} (X,{S_1},{S_2})\) is to find a basis
\({H} = (h_1, \dots, h_p)\) of \(E\) and a finite non-increasing real
sequence \(\Lambda = (\lambda_1 \geq \ldots \geq \lambda_p)\) such that
\begin{equation}\phantomsection\label{eq-icsdef}{
\operatorname{ICS} (X,{S_1},{S_2}) : \left\{ \begin{array}{lcl}
\langle{S_1} [X] h_j, h_{j'} \rangle &= \delta_{jj'} \\
\langle{S_2} [X] h_j, h_{j'} \rangle &= \delta_{jj'} \lambda_j
\end{array} \text{ for all } 1 \leq j,j' \leq p, \right.
}\end{equation} where \(\delta_{jj'}\) equals \(1\) if \(j=j'\) and
\(0\) otherwise. Such a basis \(H\) is called an
\(\operatorname{ICS} (X,{S_1},{S_2})\) eigenbasis, whose elements are
\(\operatorname{ICS} (X,{S_1},{S_2})\) eigenobjects. Such a \(\Lambda\)
is called an \(\operatorname{ICS} (X,{S_1},{S_2})\) spectrum, whose
elements are called \(\operatorname{ICS} (X,{S_1},{S_2})\) eigenvalues
or generalised kurtosis. Given an \(\operatorname{ICS} (X,{S_1},{S_2})\)
eigenbasis \(H\) and \(1 \leq j \leq p\), the real number
\begin{equation}\phantomsection\label{eq-ic}{z_j = \langle X - \mathbb EX, h_j \rangle}\end{equation}
is called the \(j\)-th invariant coordinate (in the eigenbasis \(H\)).

\end{definition}

In Definition~\ref{def-ics}, our coordinate emancipation procedure does
not yield a generalisation to infinite-dimensional Hilbert spaces, where
a basis \(H\) would not be properly defined as it is not necessarily
orthonormal.

\begin{remark}[Multivariate case]
If \(E = \mathbb R^p\), we identify \(S_1\) and \(S_2\) with their
associated \((p\times p)\)-matrices in the canonical basis, and we
identify an ICS eigenbasis \(H\) with the \((p\times p)\)-matrix of its
vectors stacked column-wise, so that we retrieve the classical
formulation of invariant coordinate selection by
\textcite{tyler_invariant_2009}.
\end{remark}

In the ICS problem Equation~\ref{eq-icsdef}, the scatter operators
\(S_1\) and \(S_2\) do not play symmetrical roles. This is because the
usual method of solving \(\operatorname{ICS} (X,{S_1},{S_2})\) is to use
the associated inner product of \(S_1 [X]\), which requires \(S_1 [X]\)
to be injective. In that case, Proposition~\ref{prp-existence} in the
Appendix proves the existence of solutions to the ICS problem.

Another way to understand the coordinate-free nature of this ICS problem
is to work with data isometrically represented in two spaces and to
understand how we can relate a given ICS problem in the first space to a
corresponding ICS problem in the second. This is the object of the
following proposition, which will be used in Section~\ref{sec-icscoord}.

\begin{proposition}[]\protect\hypertarget{prp-isometry}{}\label{prp-isometry}

Let
\(\varphi : (E, \langle \cdot, \cdot \rangle_E) \rightarrow (F, \langle \cdot, \cdot \rangle_F)\)
be an isometry between two Euclidean spaces of dimension \(p\),
\(\mathcal E \subseteq L^1 (\Omega, E)\) an affine invariant set of
integrable \(E\)-valued random objects, \(S_1^{\mathcal E}\) and
\(S_2^{\mathcal E}\) two affine equivariant scatter operators on
\(\mathcal E\). Then:

\begin{enumerate}
\def\labelenumi{(\alph{enumi})}
\item
  \(\mathcal F = \varphi (\mathcal E) = \{ \varphi (X^{\mathcal E}), X^{\mathcal E} \in \mathcal E \}\)
  is an affine invariant set of integrable \(F\)-valued random objects,
  and we denote
  \(X^{\mathcal F} = \varphi(X^{\mathcal E}) \in \mathcal F\) whenever
  \(X^{\mathcal E} \in \mathcal E\);
\item
  \(S_\ell^{\mathcal F} :  X^{\mathcal F} \in \mathcal F \mapsto \varphi \circ S_\ell^{\mathcal E} [X^{\mathcal E}] \circ \varphi^{-1}, \ell \in \{1,2\},\)
  are two affine equivariant scatter operators on \(\mathcal F\);
\item
  \(H^{\mathcal F} = \varphi(H^{\mathcal E}) = (\varphi (h_1^{\mathcal E}), \dots, \varphi (h_p^{\mathcal E}))\)
  is a basis of \(F\) whenever
  \(H^{\mathcal E}  = (h_1^{\mathcal E}, \dots, h_p^{\mathcal E})\) is a
  basis of \(E\).
\end{enumerate}

For any \(E\)-valued random object \(X^{\mathcal E} \in \mathcal E\),
any basis
\(H^{\mathcal E} = (h_1^{\mathcal E}, \dots, h_p^{\mathcal E})\) of
\(E\), and any finite non-increasing real sequence
\(\Lambda = (\lambda_1 \geq \ldots \geq \lambda_p)\) the following
assertions are equivalent:

\begin{enumerate}
\def\labelenumi{(\roman{enumi})}
\item
  \((H^{\mathcal E}, \Lambda)\) solves
  \(\operatorname{ICS} (X^{\mathcal E}, S_1^{\mathcal E}, S_2^{\mathcal E})\)
  in the space \(E\)
\item
  \((H^{\mathcal F}, \Lambda)\) solves
  \(\operatorname{ICS} (X^{\mathcal F}, S_1^{\mathcal F}, S_2^{\mathcal F})\)
  in the space \(F\).
\end{enumerate}

\end{proposition}

\subsection{The case of weighted covariance
operators}\label{sec-weighted}

A difficulty in ICS is to find interesting scatter operators that
capture the non-ellipticity of the random object. Usually, for
multivariate data, we use the pair of scatter matrices
\((\operatorname{Cov}, \operatorname{Cov}_4)\). In this section, we
define an important family of scatter operators, namely the weighted
covariance operators, which contains both \(\operatorname{Cov}\) and
\(\operatorname{Cov}_4\). They are explicitly defined by coordinate-free
formulas which allow us to relate ICS problems using weighted covariance
operators between any two Euclidean spaces. We denote by
\(\mathcal{GL} (E)\) the group of linear automorphisms of \(E\) and by
\(A^{1/2}\) the unique non-negative square root of a linear mapping
\(A\).

\begin{definition}[Weighted covariance
operators]\protect\hypertarget{def-covw}{}\label{def-covw}

For any measurable function \(w: \mathbb R^+ \rightarrow \mathbb R\),
let \[
\mathcal E_w = \left\{ X \in L^2 (\Omega, E) \left| \, \operatorname{Cov}[X] \in \mathcal{GL}(E) \\
\text{ and } w \left( \left\| \operatorname{Cov}[X]^{-1/2} (X - \mathbb EX) \right\| \right) \| X - \mathbb EX \| \in L^2 (\Omega, \mathbb R) \right. \right\}.
\] Note that \(\mathcal E_w\) is an affine invariant set of integrable
\(E\)-valued random objects. For \(X \in \mathcal E_w\), we define the
\(w\)-weighted covariance operator \(\operatorname{Cov}_w [X]\) by
\begin{equation}\phantomsection\label{eq-covw}{
\forall (x,y) \in E^2, \langle \operatorname{Cov}_w [X] x, y \rangle =
\mathbb E \left[ w^2 \left( \left\| \operatorname{Cov}[X]^{-1/2} (X - \mathbb EX) \right\| \right) \langle X - \mathbb EX, x \rangle \langle X - \mathbb EX, y \rangle \right].
}\end{equation} When necessary, we will also write
\(\operatorname{Cov}_w^E\) for the \(w\)-weighted covariance operator on
\(E\) to avoid any ambiguity. It is easy to check that weighted
covariance operators are affine equivariant scatter operators in the
sense of Definition~\ref{def-scatter}.

\end{definition}

\begin{example}[]\protect\hypertarget{exm-cov}{}\label{exm-cov}

If \(w=1\), we retrieve \(\operatorname{Cov}\), the usual covariance
operator on \(L^2 (\Omega, E)\).

\end{example}

\begin{example}[]\protect\hypertarget{exm-cov4}{}\label{exm-cov4}

If for \(x \in \mathbb R^+\), \(w(x) = (p+2)^{-1/2} x\), we obtain the
fourth-order moment operator \(\operatorname{Cov}_4\) \autocite[as in][
for the case \(E = \mathbb R^p\)]{nordhausen_usage_2022} on
\(\mathcal E_w = \left\{ X \in L^4 (\Omega, E) \, | \, \operatorname{Cov}[X] \in \mathcal{GL} (E) \right\}\).

\end{example}

The following corollary applies Proposition~\ref{prp-isometry} to the
pair of \(w_\ell\)-weighted covariance operators
\(S_\ell^{\mathcal E} = \operatorname{Cov}_{w_\ell}, \ell \in  \{ 1, 2 \}\),
for which the corresponding \(S_\ell^{\mathcal F}\) are exactly the
\(w_\ell\)-weighted covariance operators on \(F\).

\begin{corollary}[]\protect\hypertarget{cor-isometry-covw}{}\label{cor-isometry-covw}

Let
\((E, \langle \cdot, \cdot \rangle_E) \overset{\varphi}{\rightarrow} (F, \langle \cdot, \cdot \rangle_F)\)
be an isometry between two Euclidean spaces of dimension \(p\) and
\(w_1, w_2 : \mathbb R^+ \rightarrow \mathbb R\) two measurable
functions. For any integrable \(E\)-valued random object
\(X \in \mathcal E_{w_1} \cap \mathcal E_{w_2}\) (with the notations
from Definition~\ref{def-covw}), the equality
\begin{equation}\phantomsection\label{eq-isometry-covw}{\operatorname{Cov}_{w_\ell}^F [\varphi(X)] = \varphi \circ \operatorname{Cov}_{w_\ell}^E [X] \circ \varphi^{-1}}\end{equation}
holds for \(\ell \in  \{ 1, 2 \}\), as well as the equivalence between
the following assertions, for any basis \(H = (h_1, \dots, h_p)\) of
\(E\), and any finite non-increasing real sequence
\(\Lambda = (\lambda_1 \geq \ldots \geq \lambda_p)\):

\begin{enumerate}
\def\labelenumi{(\roman{enumi})}
\item
  \((H, \Lambda)\) solves
  \(\operatorname{ICS} (X, \operatorname{Cov}_{w_1}^E, \operatorname{Cov}_{w_2}^E)\)
  in the space \(E\).
\item
  \((\varphi(H), \Lambda)\) solves
  \(\operatorname{ICS} (\varphi(X), \operatorname{Cov}_{w_1}^F, \operatorname{Cov}_{w_2}^F)\)
  in the space \(F\).
\end{enumerate}

\end{corollary}

\subsection{Implementation}\label{sec-icscoord}

In order to implement coordinate-free ICS in any Euclidean space \(E\),
we restrict our attention to the pair
\((\operatorname{Cov}_{w_1}, \operatorname{Cov}_{w_2})\) of weighted
covariance operators defined in Section~\ref{sec-weighted}. Note that we
could also transport other known scatter matrices, such as the Minimum
Covariance Determinant \autocite[defined
in][]{rousseeuw_multivariate_1985}, back to the space \(E\) using
Proposition~\ref{prp-isometry}, but this approach would no longer be
coordinate-free.

We now choose a basis \(B=(b_1, \dots, b_p)\) of \(E\) in order to
represent each element \(x\) of \(E\) by its coordinate vector
\([x]_B = ([x]_{b_1} \dots [x]_{b_p})^\top \in \mathbb R^p\). Then, the
following corollary of Proposition~\ref{prp-isometry} allows one to
relate the coordinate-free approach in \(E\) to three different
multivariate approaches applied to the coordinate vectors in any basis
\(B\) of \(E\), where the Gram matrix
\(G_B = (\langle b_j, b_{j'} \rangle)_{1 \leq j,j' \leq p}\) appears,
accounting for the non-orthonormality of \(B\). Notice that, since the
ICS problem has been defined in Section~\ref{sec-icsproblem} without any
reference to a particular basis, it is obvious that the basis \(B\) has
no influence on ICS.

\begin{corollary}[]\protect\hypertarget{cor-coord}{}\label{cor-coord}

Let \((E, \langle \cdot, \cdot \rangle)\) be a Euclidean space of
dimension \(p\), \(w_1, w_2 : \mathbb R^+ \rightarrow \mathbb R\) two
measurable functions. Let \(B\) be any basis of \(E\),
\(G_B = (\langle b_j, b_{j'} \rangle)_{1 \leq j,j' \leq p}\) its Gram
matrix and \([ \cdot ]_B\) the linear map giving the coordinates in
\(B\). For any \(X \in \mathcal E_{w_1} \cap \mathcal E_{w_2}\) (with
the notations from Definition~\ref{def-covw}), any basis
\({H} = (h_1, \dots, h_p)\) of \(E\), and any finite non-increasing real
sequence \(\Lambda = (\lambda_1 \geq \ldots \geq \lambda_p)\) the
following assertions are equivalent:

\begin{enumerate}
\def\labelenumi{(\arabic{enumi})}
\item
  \((H, \Lambda)\) solves
  \(\operatorname{ICS} (X, \operatorname{Cov}_{w_1}^E, \operatorname{Cov}_{w_2}^E)\)
  in the space \(E\)
\item
  \(({G_B^{1/2}} [H]_B, \Lambda)\) solves
  \(\operatorname{ICS} ({G_B^{1/2}} [X]_B, \operatorname{Cov}_{w_1}, \operatorname{Cov}_{w_2})\)
  in the space \(\mathbb R^p\)
\item
  \(([H]_B, \Lambda)\) solves
  \(\operatorname{ICS} ({G_B} [X]_B, \operatorname{Cov}_{w_1}, \operatorname{Cov}_{w_2})\)
  in the space \(\mathbb R^p\)\{\#eq-third\}
\item
  \(({G_B} [H]_B, \Lambda)\) solves
  \(\operatorname{ICS} ([X]_B, \operatorname{Cov}_{w_1}, \operatorname{Cov}_{w_2})\)
  in the space \(\mathbb R^p\)
\end{enumerate}

\emph{where \([H]_B\) denotes the non-singular \(p \times p\) matrix
representing the basis \(([h_1]_B, \dots, [h_p]_B)\) of
\(\mathbb R^p\).}

\end{corollary}

In practice, we prefer Assertion (3) (transforming the data by the Gram
matrix of the basis) because it is the only one that does not require
inverting the Gram matrix in order to recover the eigenobjects. Then,
the problem is reduced to multivariate ICS, already implemented in the R
package \texttt{ICS} using the QR decomposition
\autocite{archimbaud_numerical_2023}. This QR approach enhances
stability compared to methods based on a joint diagonalisation of two
scatter matrices, which can be numerically unstable in some
ill-conditioned situations.

After we obtain the ICS eigenelements, we can use them to reconstruct
the original random object, in order to interpret the contribution of
each invariant component. Proposition~\ref{prp-reconstruction} in the
Appendix generalises the multivariate reconstruction formula to complex
data. In order to implement this reconstruction, we need the coordinates
of the elements of the dual ICS eigenbasis. Identifying the basis
\([H]_B\) with the matrix whose columns are its vectors, the dual basis
\([H^*]_B\) is the matrix
\[[H^*]_B = \left( [H]_B^\top G_B \right)^{-1}.\]

\begin{remark}[Empirical ICS and estimation]
In order to work with samples of complex random objects, we can study
the particular case of a finite \(E\)-valued random object \(X\) where
we have a fixed sample \(D_n=(x_1, \dots, x_n)\) and we assume that
\(X\) follows the empirical probability distribution \(P_{D_n}\) of
\((x_1, \dots, x_n)\). In that case, the expressions (in
Definition~\ref{def-covw}) for instance) of the form \(\mathbb E f(X)\)
for any function \(f\) are discrete and equal to
\(\frac1n \sum_{i=1}^n f(x_i)\).

Now, let us assume that we observe an i.i.d.~sample
\(D_n = (X_1, \dots, X_n)\) following the distribution of an unknown
\(E\)-valued random object \(X_0\). We can estimate solutions of the
problem \(\operatorname{ICS} (X_0, S_1, S_2)\) from
Definition~\ref{def-ics} by working conditionally on the data
\((X_1, \dots, X_n)\) and taking the particular case where \(X\) follows
the empirical probability distribution \(P_{D_n}\). This defines
estimates of the \(\operatorname{ICS} (X_0, S_1, S_2)\) eigenobjects as
solutions of an ICS problem involving empirical scatter operators. Since
the population version of ICS for a complex random object \(X \in E\) is
more concise than its sample counterpart for
\(D_n = (X_1, \dots, X_n)\), we shall use the notations of the former in
the next sections.
\end{remark}

\subsection{ICS for compositional data}\label{sec-icscoda}

The specific case of coordinate-free ICS for compositional data is
equivalent to the approach of \textcite{ruiz-gazen_detecting_2023}. To
see this, let us consider the simplex
\(E = (\mathcal S^{p+1}, \oplus, \odot, \langle \cdot, \cdot \rangle_{\mathcal S^{p+1}})\)
of dimension \(p\) in \(\mathbb R^{p+1}\) with the Aitchison structure
\autocite{pawlowskyglahn_modeling_2015}. The results from 5.1
(resp.~5.2) in \autocite{ruiz-gazen_detecting_2023} can be recovered by
applying Corollary~\ref{cor-isometry-covw} to any isometric log-ratio
transformation \autocite[see][ for a
definition]{pawlowskyglahn_modeling_2015} (resp.~the centred log-ratio
transformation).

Corollary~\ref{cor-coord} gives a new characterisation of the problem
\(\operatorname{ICS} (X, \operatorname{Cov}_{w_1}, \operatorname{Cov}_{w_2})\)
using additive log-ratio transformations. For a given index
\(1 \leq j \leq p\), let \(B_j = (b_1, \dots, b_p)\) denote the basis of
\(\mathcal S^{p+1}\) corresponding to the \(\operatorname{alr}_j\)
transformation, i.e.~obtained by taking the canonical basis of
\(\mathbb R^{p+1}\), removing the \(j\)-th vector and applying the
exponential. In that case, it is easy to compute the \(p \times p\) Gram
matrix of \(B_j\):
\[G_{B_j} = I_p - \frac1{p+1} \mathbf 1_p \mathbf 1_p ^\top = \begin{pmatrix}
        1 - \frac1{p+1} & -\frac1{p+1} & \dots & -\frac1{p+1} \\
        -\frac1{p+1} & \ddots & \ddots & \vdots \\
        \vdots & \ddots & \ddots & - \frac1{p+1} \\
        -\frac1{p+1} & \dots & -\frac1{p+1} & 1 - \frac1{p+1}
    \end{pmatrix}.\] Then, we get the equivalence between the following
two ICS problems:

\begin{enumerate}
\def\labelenumi{\arabic{enumi}.}
\item
  \((H, \Lambda)\) solves
  \(\operatorname{ICS} (X, \operatorname{Cov}_{w_1}, \operatorname{Cov}_{w_2})\)
  in the space \(\mathcal S^{p+1}\)
\item
  \((\operatorname{alr}_j (H), \Lambda)\) solves
  \(\operatorname{ICS} (\operatorname{clr} (X)^{(j)}, \operatorname{Cov}_{w_1}, \operatorname{Cov}_{w_2})\)
  in the space \(\mathbb R^p\)
\end{enumerate}

where
\(\operatorname{clr} (x)^{(j)} = G_{B_j} \operatorname{alr}_j (x)\) is
the centred log-ratio transform of \(x \in \mathcal S^{p+1}\) from which
the \(j\)-th coordinate has been removed. This suggests a new and
fastest implementation of invariant coordinate selection for
compositional data, in an unconstrained space and only requiring the
choice of an index \(j\) instead of a full contrast matrix.

\subsection{ICS for functional data}\label{sec-icsfun}

The difficulty of functional data (in the broader sense, encompassing
density data) is twofold: first, functions are usually analysed within
the infinite-dimensional Hilbert space \(L^2 (a,b)\), second, a random
function is almost never observed for every argument, but rather on a
discrete grid. This grid can be regular or irregular, deterministic or
random, dense (the grid spacing goes to zero) or sparse. We describe a
general framework for adapting coordinate-free ICS to functional data,
solving both difficulties at the same time by smoothing the observed
values into a random function \(u\) that belongs to a Euclidean subspace
\(E\) of \(L^2 (a,b)\).

\subsubsection{Choosing an approximating Euclidean
subspace}\label{choosing-an-approximating-euclidean-subspace}

We usually choose polynomial spaces, spline spaces with given knots and
order, or spaces spanned by a truncated Hilbert basis of \(L^2 (a,b)\).
In practice, this choice also depends on the preprocessing method that
we have in mind to smooth discrete observations into functions.

\subsubsection{Preprocessing the observations into the approximating
space}\label{preprocessing-the-observations-into-the-approximating-space}

Considering a dense, deterministic grid \((t_1, \dots, t_N)\), we need
to reconstruct an \(E\)-valued random function \(u\) from its noisy
observed values
\((u(t_1) + \varepsilon_1, \dots, u(t_N) + \varepsilon_N)\). There are
many well-documented approximation techniques to carry out this
preprocessing step, such as interpolation, spline smoothing, or Fourier
methods \autocite[for a detailed presentation,
see][]{eubank_nonparametric_2014}.

\subsubsection{Solving ICS in the approximating
space}\label{solving-ics-in-the-approximating-space}

Once we have obtained an \(E\)-valued random function \(u\), we can
apply the method described in Section~\ref{sec-icscoord} to reduce
\(\operatorname{ICS} (u, \operatorname{Cov}_{w_1}, \operatorname{Cov}_{w_2})\)
to a multivariate problem on the coordinates in a basis of \(E\). In
particular, for an orthonormal basis \(B\) of \(E\) (such as a Fourier
basis or a Hermite polynomial basis), Corollary~\ref{cor-coord} gives
the equivalence between the following two assertions:

\begin{enumerate}
\def\labelenumi{\arabic{enumi}.}
\item
  \((H, \Lambda)\) solves
  \(\operatorname{ICS} (u, \operatorname{Cov}_{w_1}, \operatorname{Cov}_{w_2})\)
  in the space \(E\)
\item
  \(([H]_B, \Lambda)\) solves
  \(\operatorname{ICS} ([u]_B, \operatorname{Cov}_{w_1}, \operatorname{Cov}_{w_2})\)
  in the space \(\mathbb R^p\).
\end{enumerate}

If \(E\) is a finite-dimensional spline space, we usually work with the
coordinates of \(u\) in a B-spline basis of \(E\), but then we should
take into account its Gram matrix, as in Corollary~\ref{cor-coord}.\\
ICS has previously been defined for multivariate functional data by
\textcite{archimbaud_ics_2022}, who define a pointwise method and a
global method. Unlike the pointwise approach, which is specific to
multivariate functional data, the global method can also be applied to
univariate functional data in \(L^2 (a,b)\), as it corresponds to
applying multivariate ICS to truncated coordinate vectors in a Hilbert
basis of \(L^2 (a,b)\). The above framework retrieves the global method
in \autocite{archimbaud_ics_2022} as a particular case when taking a
Hilbert basis \(B\) of \(L^2 (a,b)\) and solving coordinate-free ICS in
the space \(E\) spanned by the \(p\) first elements of \(B\).

\subsection{ICS for distributional data}\label{sec-icsdens}

A first option to adapt ICS to density data would be to consider it as
constrained functional data and directly follow the approach of
Section~\ref{sec-icsfun}. However, distributional data does not reduce
to density data \autocite[such as absorbance spectra studied
in][]{ferraty_functional_2002}, as it can also be histogram data or
sample data (such as the dataset of temperature samples analysed in
Section~\ref{sec-appliout}). Moreover, the framework of Bayes Hilbert
spaces, described by \autocite{van_den_boogaart_bayes_2014} and recalled
in the Appendix, is specifically adapted to the study of distributional
data. Taking into account the infinite-dimensional nature of
distributional data, we follow a similar framework as the one of Section
Section~\ref{sec-icsfun}, restricting our attention to
finite-dimensional subspaces \(E\) of the Bayes space \(B^2(a,b)\) with
the Lebesgue measure as reference.

\subsubsection{Choosing an approximating Euclidean
space}\label{choosing-an-approximating-euclidean-space}

Following smoothing splines methods, adapted to Bayes spaces by
\textcite{machalova_preprocessing_2016} and recalled in the Appendix, we
choose to work in the space \(E = \mathcal C^{\Delta \gamma}_d (a,b)\)
of compositional splines on \((a,b)\) of order \(d+1\) with knots
\(\Delta \gamma = (\gamma_1, \dots, \gamma_k)\). Note that the centred
log-ratio transform \(\operatorname{clr}\) is an isometry between \(E\)
and the space \(F = \mathcal Z^{\Delta \gamma}_d (a,b)\) of
zero-integral splines on \((a,b)\) of order \(d+1\) (degree less than or
equal to \(d\)) and with knots
\(\Delta \gamma = (\gamma_1, \dots, \gamma_k)\). They both have
dimension \(p = k + d\).

\subsubsection{Preprocessing the observations into the approximating
space}\label{sec-preproc}

We consider the special cases of histogram data and of sample data. In
the former, we follow \autocite{machalova_compositional_2021} to smooth
each histogram into a compositional spline in \(E\). In the latter, we
assume that a random density is observed through a finite random sample
\((X_1, \dots, X_N)\) drawn from it. The preprocessing step consists in
estimating the density from the observed sample. To perform the
estimation, we need a nonparametric estimation procedure that yields a
compositional spline belonging to \(E\). That is why we opt for maximum
penalised likelihood (MPL) density estimation, introduced by
\textcite{silverman_estimation_1982}. The principle of MPL is to
maximise a penalised version of the log-likelihood over an
infinite-dimensional space of densities without parametric assumptions.
The penalty is the product of a smoothing parameter \(\lambda\) by the
integral over the interval of interest of the square of the \(m\)-th
derivative of the log density. Therefore, the objective functional is a
functional of the log density. Due to the infinite dimension of the
ambient space, the likelihood term alone is unbounded above, hence the
penalty term is necessary. In our case of densities on an interval
\((a,b)\), we select the value \(m=3\) so that \autocite[according to][
Theorem 2.1]{silverman_estimation_1982} when the smoothing parameter
tends to infinity, the estimated density converges to the parametric
maximum likelihood estimate in the exponential family of densities whose
logarithm is a polynomial of degree less than or equal to 2. This family
comprises the uniform density, exponential and Gaussian densities
truncated to \((a,b)\). In order to use MPL in \(B^2(a,b)\), we need to
add extra smoothness conditions and therefore we restrict attention to
the densities of \(B^2(a,b)\) whose log belongs to the Sobolev space of
order \(m\) on \((a,b)\), thus ensuring the existence of the penalty
term. Note that compositional splines verify these conditions. With
Theorem 4.1 in \autocite{silverman_estimation_1982}, the optimisation
problem has at least a solution. Since the estimate \(f\) of the density
of \((X_1, \dots, X_N)\) needs to belong to the chosen
finite-dimensional subspace \(E = \mathcal C^{\Delta \gamma}_d (a,b)\),
we restrict MPL to \(E\), using the R function \texttt{fda::density.fd},
designed by \textcite{ramsay_fda_2024}. This function returns the
coordinates of \(\log (f)\) in the B-spline basis with knots
\(\Delta \gamma\) and order \(d+1\), that we project onto
\(\mathcal Z^{\Delta \gamma}_d (a,b)\) and to which we apply
\(\operatorname{clr}^{-1}\) so that we obtain an element of
\(\mathcal C^{\Delta \gamma}_d (a,b)\).

\subsubsection{Solving ICS in the approximating
space}\label{solving-ics-in-the-approximating-space-1}

We have now obtained an \(E\)-valued random compositional spline \(f\).
In order to work with two weighted covariance operators
\(\operatorname{Cov}_{w_1}\) and \(\operatorname{Cov}_{w_2}\), where
\(w_1, w_2: \mathbb R^+ \rightarrow \mathbb R\) are two measurable
functions, we assume that
\(f \in \mathcal E_{w_1} \cap \mathcal E_{w_2}\), using the notations of
Definition~\ref{def-covw}. Now, we refer to Section~\ref{sec-icscoord}
to reduce the problem
\(\operatorname{ICS} (f, \operatorname{Cov}_{w_1}, \operatorname{Cov}_{w_2})\)
to a multivariate ICS problem on the coordinates of \(f\) in the
CB-spline basis of \(\mathcal C^{\Delta \gamma}_d (a,b)\)
\autocite[defined in][]{machalova_compositional_2021}, transformed by
the Gram matrix of said CB-spline basis. Note that
Corollary~\ref{cor-isometry-covw} applied to the centred log-ratio
isometry between \(\mathcal C^{\Delta \gamma}_d (a,b)\) and
\(\mathcal Z^{\Delta \gamma}_d (a,b)\) gives the equivalence between:

\begin{enumerate}
\def\labelenumi{\arabic{enumi}.}
\item
  \((H, \Lambda)\) solves
  \(\operatorname{ICS} (f, \operatorname{Cov}_{w_1}, \operatorname{Cov}_{w_2})\)
  in the space \(E = \mathcal C^{\Delta \gamma}_d (a,b)\)
\item
  \((\operatorname{clr} (H), \Lambda)\) solves
  \(\operatorname{ICS} (\operatorname{clr} (f), \operatorname{Cov}_{w_1}, \operatorname{Cov}_{w_2})\)
  in the space \(F = \mathcal Z^{\Delta \gamma}_d (a,b)\).
\end{enumerate}

Then, it is completely equivalent, and useful for implementation, to
work with the coordinates of \(\operatorname{clr} (f)\) in the ZB-spline
basis of \(\mathcal Z^{\Delta \gamma}_d (a,b)\).

\section{Outlier detection for complex data using
ICS}\label{sec-icscomplexout}

\subsection{Implementation of ICS on complex data for outlier
detection}\label{sec-icsout}

We propose using ICS to detect outliers in complex data, specifically in
scenarios with a small proportion of outliers (typically 1 to 2\%). For
this, we follow the three-step procedure defined by
\textcite{archimbaud_ics_2018}, modifying the first step based on the
implementation of coordinate-free ICS in Section~\ref{sec-icscoord}.

\subsubsection{Computing the invariant
coordinates}\label{computing-the-invariant-coordinates}

For the scatter operators, we follow the recommendation of
\textcite{archimbaud_ics_2018} who compare several pairs of more or less
robust scatter estimators in the context of a small proportion of
outliers, and conclude that
\((\operatorname{Cov}, \operatorname{Cov}_4)\) is the best choice. Thus,
we use the empirical scatter pair
\((\operatorname{Cov}, \operatorname{Cov}_4)\) (see
Example~\ref{exm-cov} and Example~\ref{exm-cov4}), and compute the
eigenvalues \(\lambda_1 \geq \ldots \geq \lambda_p\), and the invariant
coordinates \(z_{ji}, 1 \leq j \leq p\), for each observation
\(X_i, 1 \leq i \leq n\). As outlined in Section~\ref{sec-icscoord}, for
a given sample of random complex objects \(D_n= \{X_1, \dots, X_n\}\) in
a Euclidean space \(E\), solving the empirical version of ICS is
equivalent to solving an ICS problem in a multivariate framework
\autocite[see][]{tyler_invariant_2009} with the coordinates of the
objects in a basis \(B\) of \(E\). In order to choose a basis, we follow
the specific recommendations for each type of data from
Section~\ref{sec-icscoda} and Section~\ref{sec-icsdens}.

\subsubsection{Selecting the invariant
components}\label{selecting-the-invariant-components}

The second step of the outlier detection procedure based on ICS is the
selection of the \(\kappa<p\) relevant invariant components and the
computation of the ICS distances. For each of the \(n\) observations,
the ICS distance is equal to the Euclidean norm of the reconstructed
data using the \(\kappa\) selected invariant components. In the case of
a small proportion of outliers and for the scatter pair
\((\operatorname{Cov}, \operatorname{Cov}_4)\), the invariant components
of interest are associated with the largest eigenvalues and the squared
ICS distances are equal to \(\displaystyle \sum_{j=1}^\kappa z_{ji}^2\).
As noted by \textcite{archimbaud_ics_2018}, there exist several methods
for the selection of the number of invariant components. One approach is
to examine the scree plot, as in PCA. This method, recommended by
\textcite{archimbaud_ics_2018}, is not automatic. Alternative automatic
selection methods apply univariate normality tests to each component,
starting with the first one, and using some Bonferroni correction
\autocite[for further details see page 13 of][]{archimbaud_ics_2018}. In
the present paper, we use the scree plot approach when there is no need
of an automatic method, and we use the D'Agostino normality test for
automatic selection. The level for the first test (before Bonferroni
correction) is 5\%. Dimension reduction involves retaining only the
first \(\kappa\) components of ICS instead of the original \(p\)
variables. Note that when all the invariant components are retained, the
ICS distance is equal to the Mahalanobis distance.

\subsubsection{Choosing a cut-off}\label{choosing-a-cut-off}

The computation of ICS distances allows to rank the observations in
decreasing order, with those having the largest distances potentially
being outliers. However, in order to identify the outlying densities, we
need to define a cut-off, and this constitutes the last step of the
procedure. Following \textcite{archimbaud_ics_2018}, we derive cut-offs
based on Monte Carlo simulations from the standard Gaussian
distribution. For a given sample size and number of variables, we
generate 10,000 standard Gaussian samples and compute the empirical
quantile of order 97.5\% of the ICS-distances using the three steps
previously described. An observation with an ICS distance larger than
this quantile is flagged as an outlier.\\
The procedure described above has been illustrated in several examples
\autocite[see][]{archimbaud_ics_2018}, and is implemented in the R
package \texttt{ICSOutlier}
\autocite[see][]{nordhausen_icsoutlier_2023}. However, in the context of
densities, the impact of preprocessing parameters (see
Section~\ref{sec-icsdens}) on the ICSOutlier procedure emerges as a
crucial question that needs to be examined.

\subsection{Influence of the preprocessing parameters for the density
data application}\label{sec-toyex}

As a toy example, consider the densities of the maximum daily
temperatures for the 26 provinces of the two regions Red River Delta and
Northern Midlands and Mountains in Northern Vietnam between 2013 and
2016. We augment this data set made of 104 densities by adding the
provinces AN GIANG and BAC LIEU from Southern Vietnam in the same time
period. The total number of observations is thus 112. Details on the
original data and their source are provided in
Section~\ref{sec-datades}.

\phantomsection\label{cell-fig-vnt_north_south_map}
\begin{figure}[H]

\centering{

\pandocbounded{\includegraphics[keepaspectratio]{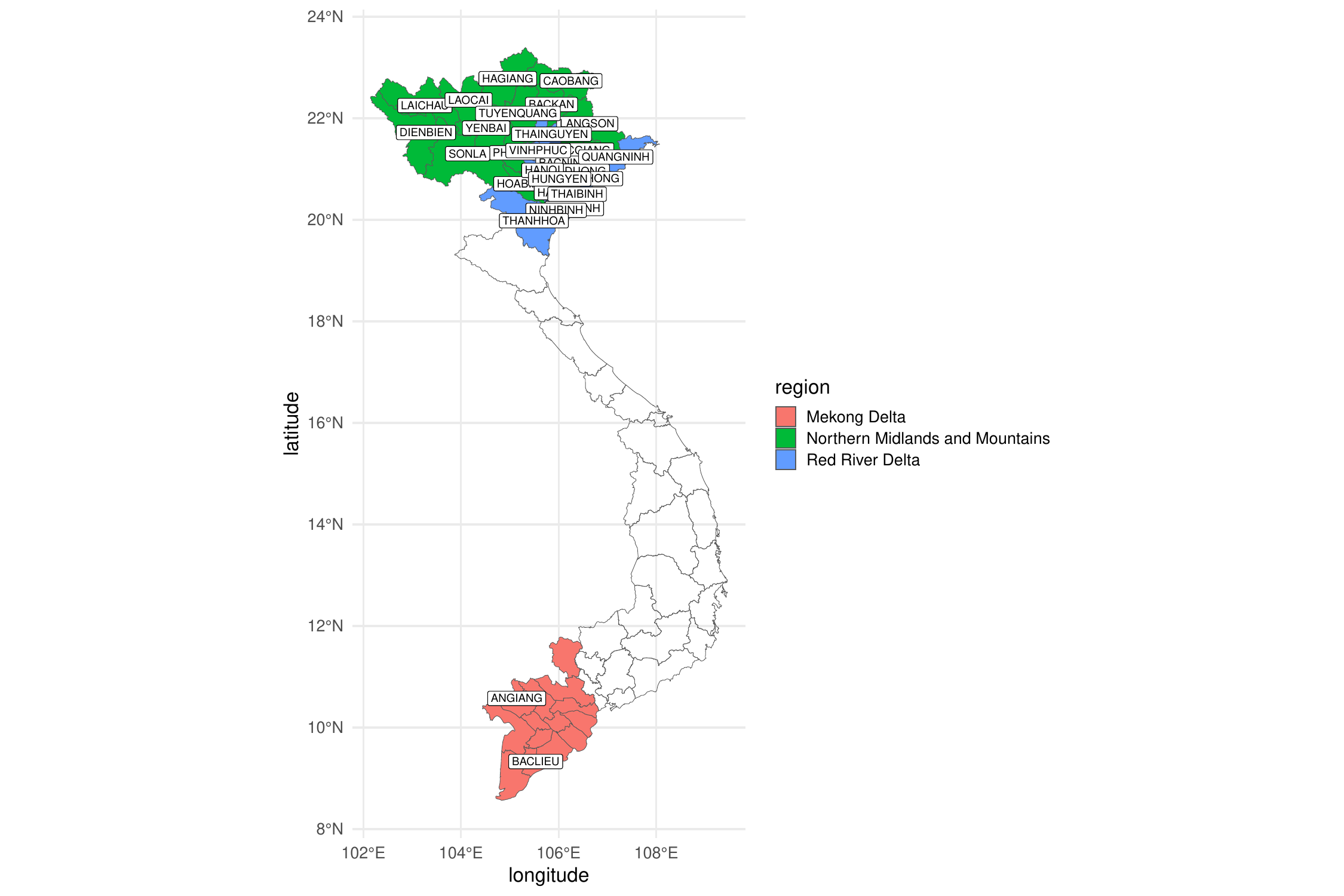}}

}

\caption{\label{fig-vnt_north_south_map}Map of Vietnam showing the 63
provinces, with the three regions under study colour-coded. The 28
provinces included in the toy example are labelled.}

\end{figure}%

\begin{figure}

\begin{minipage}{0.50\linewidth}

\centering{

\pandocbounded{\includegraphics[keepaspectratio]{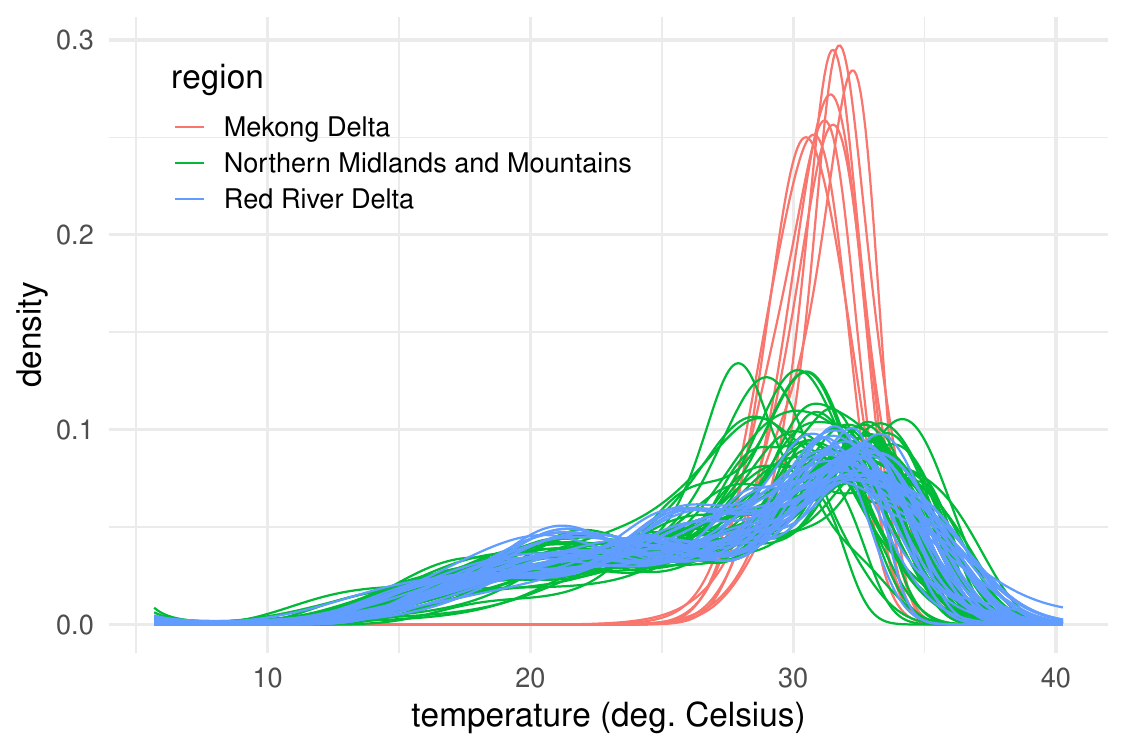}}

}

\subcaption{\label{fig-vnt_north_south_densclr-1}}

\end{minipage}%
\begin{minipage}{0.50\linewidth}

\centering{

\pandocbounded{\includegraphics[keepaspectratio]{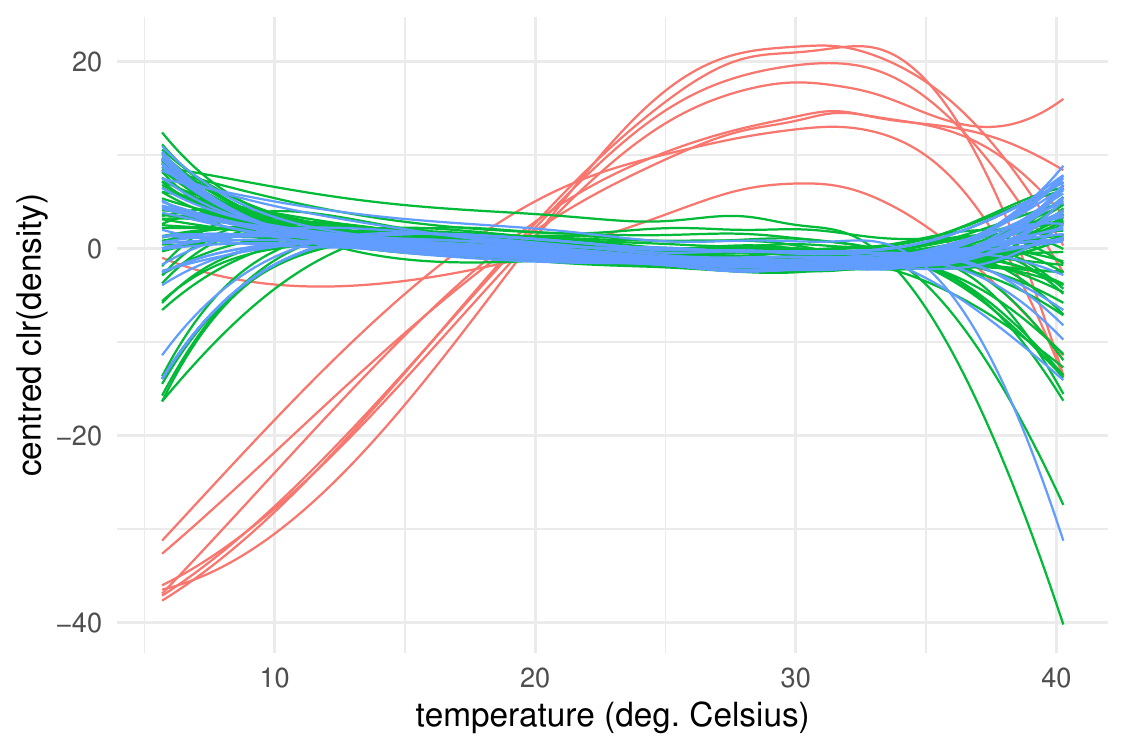}}

}

\subcaption{\label{fig-vnt_north_south_densclr-2}}

\end{minipage}%

\caption{\label{fig-vnt_north_south_densclr}Plots of the 112 densities
(left panel) and clr densities (right panel), colour-coded by region for
the toy example.}

\end{figure}%

Figure~\ref{fig-vnt_north_south_map} displays a map of Vietnam with the
contours of all provinces and coloured according to their administrative
region, allowing the reader to locate the 26 provinces in the North and
the two in the South. As shown on the left panel of
Figure~\ref{fig-vnt_north_south_densclr}, the eight densities of the two
provinces from the South for the four years exhibit a very different
shape (in red) compared to the northern provinces (in blue and green),
with much more concentrated maximum temperatures. These two provinces
should be detected as outliers when applying the ICSOutlier methodology.
However, the results may vary depending on the choice of preprocessing
parameters (see Section~\ref{sec-preproc}). Our goal is to analyse how
the detected outliers vary depending on the preprocessing when using the
maximum penalised likelihood method with splines of degree less than or
equal to \(d=4\). Specifically, we study the influence on the results of
ICSOutlier of the smoothing parameter \(\lambda\), the number of inside
knots \(k\), and the location of the knots defining the spline basis.

The number \(\kappa\) of selected invariant components is fixed at four
in all experiments to facilitate interpretation. This value has been
chosen after viewing the scree plots of the ICS eigenvalues following
the recommendations in Section~\ref{sec-icsout}. For each of the
experimental scenarios detailed below, we compute the squared ICS
distances of the 112 observations as defined in
Section~\ref{sec-icsout}, using \(\kappa=4\). Observations are
classified as outliers when their squared ICS distance exceeds the
threshold defined in Section~\ref{sec-icsout}, using a level of
\(2.5\%\). For each experiment, we plot on
Figure~\ref{fig-vnt_north_south_grid} the indices of the observations
from 1 to 112 on the \(y\)-axis, marking outlying observations with dark
squares. The eight densities from Southern Vietnam are in red and
correspond to indices 1 to 8. We consider the following scenarios:

\begin{itemize}
\tightlist
\item
  the knots are either located at the quantiles of the temperature
  values (top panel on Figure~\ref{fig-vnt_north_south_grid}) or equally
  spaced (bottom panel on Figure~\ref{fig-vnt_north_south_grid}),
\item
  from the left to the right of Figure~\ref{fig-vnt_north_south_grid},
  the number of knots varies from 0 to 14 by increments of 2, and then
  takes the values 25 and 35 (overall 10 different values). Note that
  when increasing the number of knots beyond 35, the code returns more
  and more errors due to multicollinearity issues and the results are
  not reported.
\item
  the base-10 logarithm of the parameter \(\lambda\) varies from -8 to 8
  with an increment of 1 on the \(x\)-axis of each plot.
\end{itemize}

Altogether we have \(2\times 10\times 17=340\) scenarios.
Figure~\ref{fig-vnt_north_south_grid_summary} is a bar plot showing the
observations indices on the \(x\)-axis and the frequency of outlier
detection across scenarios on the \(y\)-axis color-coded by region. The
eight densities from the two southern provinces (AN GIANG and BAC LIEU)
across the four years are most frequently detected as outliers, along
with the province of LAI CHAU (indices 33 to 36), which is located in a
mountainous region in northwest of Vietnam. On the original data, we can
see that the LAI CHAU province corresponds to densities with low values
for high maximum temperatures (above 35°C) coupled with relatively high
density values for maximum temperatures below 35°C. A few other
observations are detected several times as outliers, but less
frequently: indices 53 (TUYEN QUANG in 2013), 96 (QUANG NINH in 2016),
and 107 (THANH HOA in 2015).

Looking at Figure~\ref{fig-vnt_north_south_grid}, we examine the impact
of the preprocessing parameters on the detection of outlying
observations. First, note that the ICS algorithm returns an error when
the \(\lambda\) parameter is large (shown as white bands in some plots).
This is due to a multicollinearity problem. Even though the QR version
of the ICS algorithm is quite stable, it may still encounter problems
when multicollinearity is severe. Indeed, when \(\lambda\) is large, the
estimated densities converge to densities whose logarithm is a
polynomial of degree less than or equal to 2 (see details in
Section~\ref{sec-preproc}), and belongs to a 3-dimensional affine
subspace of the Bayes space, potentially with a dimension smaller than
that of the approximating spline space. If we compare the top and the
bottom plots, we do not observe large differences in the outlying
pattern, except for a few observations rarely detected as outliers.
Thus, the knot location has a rather small impact on the ICS results for
this data set. Regarding the impact of the \(\lambda\) parameter, the
outlier pattern remains relatively stable when the number of knots is
small (less than or equal to 6), especially when looking at the
densities from the south of Vietnam in red. For a large number of knots,
the observations detected as outliers vary with \(\lambda\). The number
of knots has more impact than their location or the \(\lambda\)
parameter. When the number of knots is smaller than or equal to 6
(corresponding to \(p=10\) variables), the plots are very similar.
However, as \(p\) increases, some observations from Southern Vietnam are
not detected for all \(\lambda\) values, while another density (QUANG
NINH in 2016) is detected for large \(\lambda\) values with equally
spaced knots, and to a lesser extent for knots at temperature quantiles.
In \autocite{archimbaud_ics_2022}, ICS is applied to multivariate
functional data with B-splines preprocessing. Based on their empirical
experience, the authors recommend using a dimension \(p\) (in their
case, the number of functional components times the number of B-splines
coefficients) no larger than the number of observations divided by 10.
Typically in multivariate analysis, the rule of thumb is that the
dimension should not exceed the number of observations divided by 5. For
functional or distributional data, it appears that even more
observations per variable are needed. The reason for this is not
entirely clear, but in the case of ICS, we can suspect that the presence
of multicollinearity, even approximate, degrades the results. By
increasing the number of knots, we precisely increase the
multicollinearity problem, especially for large values of \(\lambda\).

\phantomsection\label{cell-fig-vnt_north_south_grid}
\begin{figure}[H]

\centering{

\pandocbounded{\includegraphics[keepaspectratio]{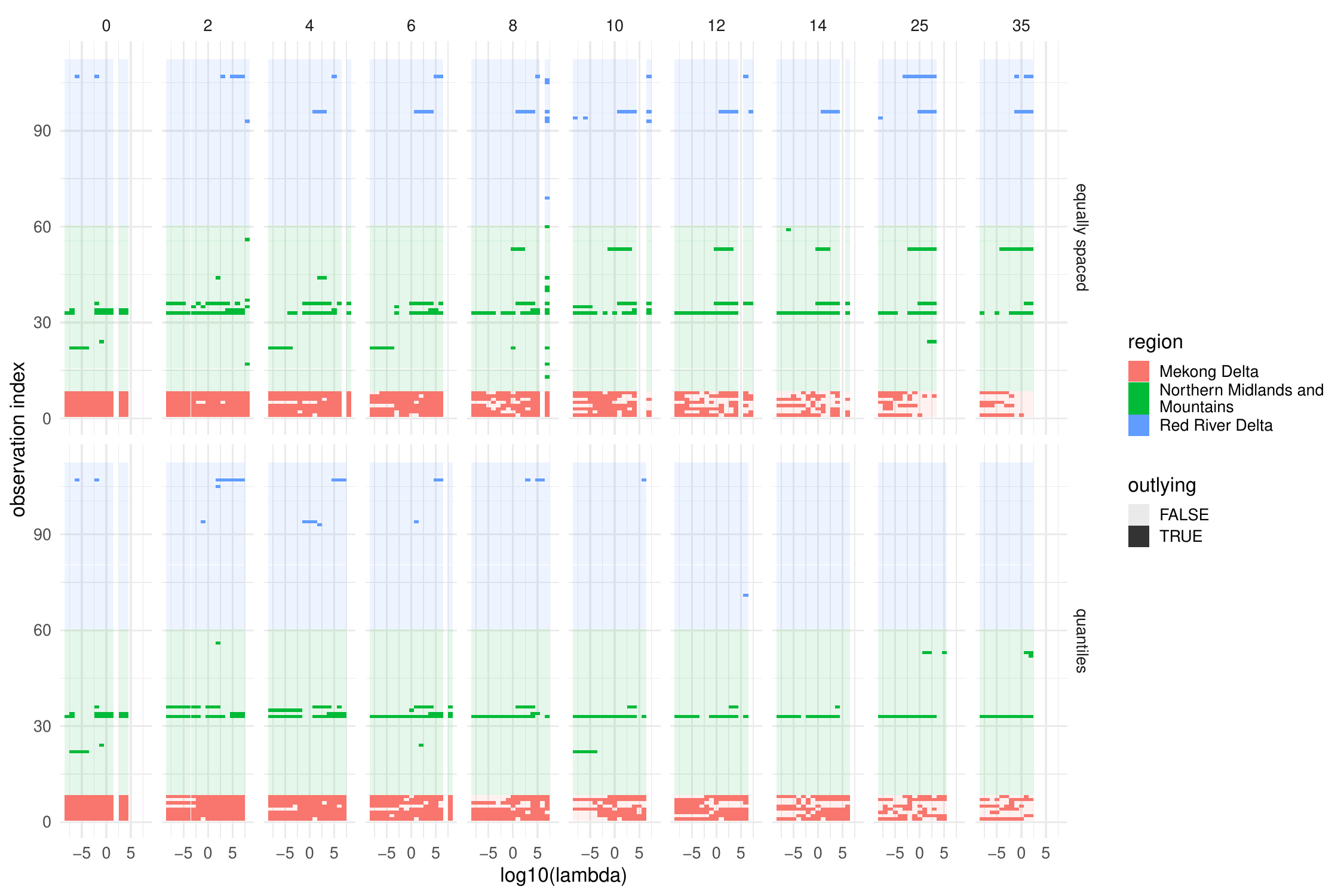}}

}

\caption{\label{fig-vnt_north_south_grid}Outlier detection by ICS across
smoothing parameters for the Vietnam toy example. \emph{Top:} knots at
quantiles; \emph{Bottom:} equally spaced knots. \emph{\(y\)-axis:}
observation indices; \emph{\(x\)-axis:} \(\lambda\) parameter. Columns
correspond to knot numbers (0-35). Outliers are dark and colour-coded by
region.}

\end{figure}%

\phantomsection\label{cell-fig-vnt_north_south_grid_summary}
\begin{figure}[H]

\centering{

\pandocbounded{\includegraphics[keepaspectratio]{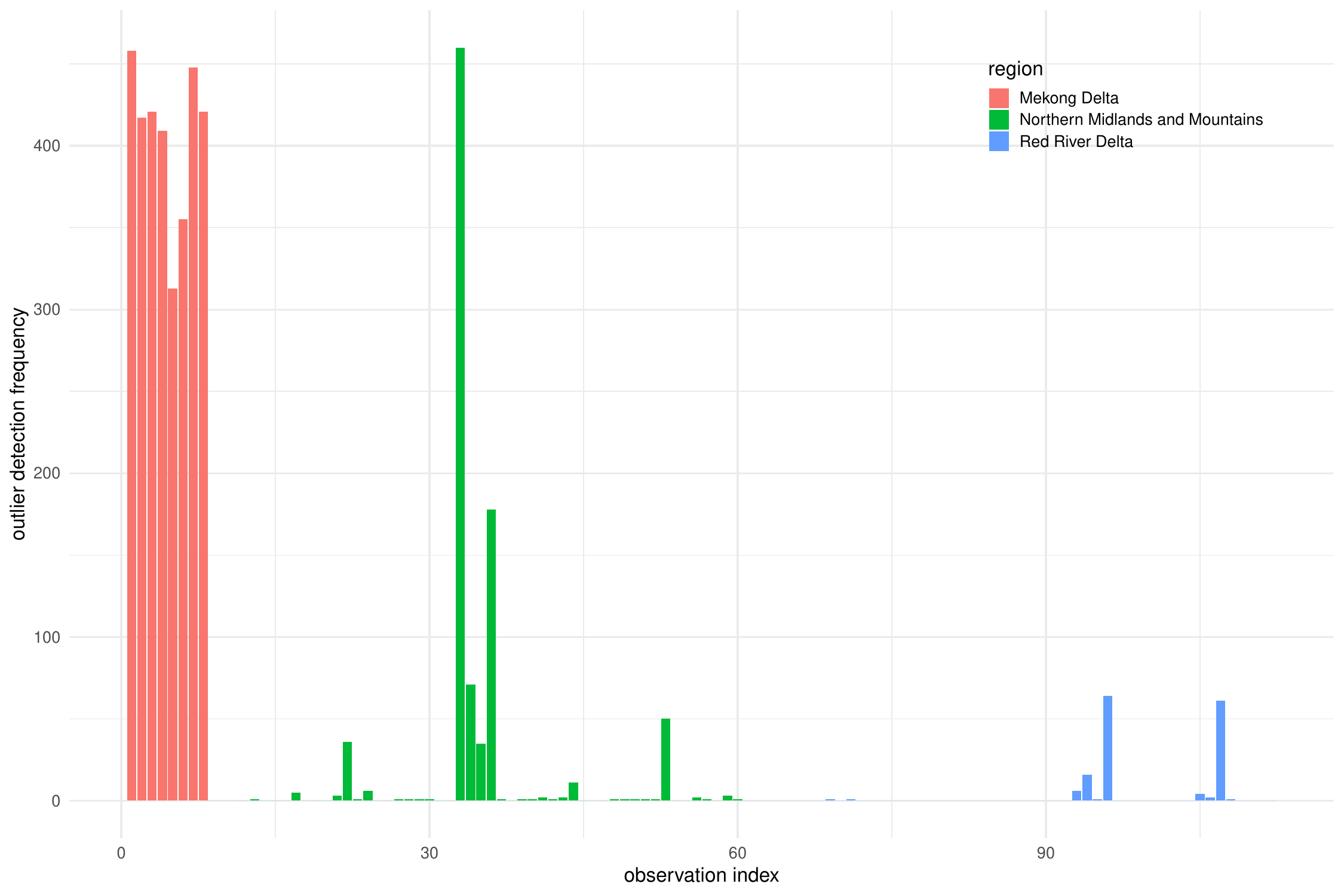}}

}

\caption{\label{fig-vnt_north_south_grid_summary}Frequency of outlier
detection by ICS across 340 scenarios with varying smoothing parameters,
for each observation in the Vietnam toy example.}

\end{figure}%

From this experimentation, we recommend using knots located at the
quantiles of the measured variable, and a number of knots such that the
number of observations is around 10 times the dimension \(p\) (here: the
dimension of the B-spline basis). The base-10 logarithm of parameter
\(\lambda\) can be chosen between -2 and 2 to avoid extreme cases and
multicollinearity problems. Moreover, the idea of launching ICS multiple
times with different preprocessing parameter values to confirm an
observation's atypical nature by its repeated detection is a strategy we
retain for real applications, as detailed in
Section~\ref{sec-app_smooth}.

\subsection{Comparison with other
methods}\label{comparison-with-other-methods}

We now compare ICS for functional data (presented in
Section~\ref{sec-icsfun}) to eight outlier detection methods already
existing in the literature, such as median-based approaches
\autocite{murph_visualisation_2024}, the modified band depth method
\autocite{sun_functional_2011} and MUOD indices
\autocite{ojo_detecting_2022}.

Our simulation uses three density-generating processes with \(2\%\) of
outliers. The scheme named \texttt{GP\_clr}, based on model 4 of the
\texttt{fdaoutlier} package \autocite[ section 4.1]{ojo_detecting_2022},
first simulates a discretised random function in \(L^2(0,1)\) from a
mixture of two Gaussian processes with different means, and applies the
inverse \(\operatorname{clr}\) transformation to obtain a random density
in the Bayes space \(B^2 (0,1)\). The scheme named \texttt{GP\_margin}
first simulates a discretised random function in \(L^2 (0,1)\) using
model 5 of the \texttt{fdaoutlier} package, which consists in a mixture
of two Gaussian processes with different covariance operators. Then, the
random density is obtained as a kind of marginal distribution of the
discrete values of the random function, where the \(x\)-axis is
discarded: theses values are considered as a random sample and smoothed
using MPL (see Section~\ref{sec-icsdens} with parameters
\(\lambda = 1\), \(10\) basis functions and knots (as well as interval
bounds) at quantiles of the full sample. This scheme is similar to the
data generating process of the Vietnamese climate dataset. Finally, the
\texttt{Gumbel} scheme first draws parameters from a mixture of two
Gaussian distributions in \(\mathbb R^2\) and computes the Gumbel
density functions corresponding to these parameters (it generates shift
outliers as described in \autocite{murph_visualisation_2024}). Note that
the output of all the schemes is a set of discretised densities on a
regular grid of size \(p=100\) that covers an interval \((a,b)\) (which
is \((0,1)\) for \texttt{GP\_clr} and \texttt{Gumbel} and the range of
the full sample for \texttt{GP\_margin}). In each sample, there are
\(n=200\) densities.

For the outlier detection methods, we denote them as
\texttt{\textless{}Approach\textgreater{}\_\textless{}Metric\textgreater{}}
so that for instance, \texttt{ICS\_B2} refers to ICS for density data in
the Bayes space \(B^2 (a,b)\). The steps of the \texttt{ICS\_B2} method
are as follows. After applying the discrete clr transformation to each
discretised density function, we approximate the underlying clr
transformed smooth density by a smoothing spline in \(L^2_0 (a,b)\)
using the preprocessing described in
\autocite{machalova_preprocessing_2016}. During this process, densities
should not take values too close to \(0\) to avoid diverging clr, so we
replace by \(10^{-8}\) all density values below this threshold. The
parameters of the compositional spline spaces are chosen by the function
\texttt{fda.usc::fdata2fd}. Then, we solve ICS in the chosen
compositional spline space, automatically selecting the components with
tests as before. The \texttt{ICS\_L2} method first smooths each
discretised density using splines in \(L^2 (a,b)\) treating the
densities as ordinary functional parameters. In the second step, we
apply ICS in the chosen spline space, selecting the components
automatically through D'Agostino normality tests. The MBD
\autocite{lopez-pintado_concept_2009} and MUOD
\autocite{azcorra_unsupervised_2018} approaches are implemented using
the package \texttt{fdaoutlier} \autocite{ojo_fdaoutlier_2023}, either
directly (\texttt{\textless{}Approach\textgreater{}\_L2}) or after
transforming the densities into log quantile densities
(\texttt{\textless{}Approach\textgreater{}\_LQD}) or into quantile
functions (\texttt{\textless{}Approach\textgreater{}\_QF}). The
median-based methods such as \texttt{Median\_LQD} and
\texttt{Median\_Wasserstein} are described in
\autocite{murph_visualisation_2024} and implemented in the
\texttt{DeBoinR} package from \autocite{murph_deboinr_2023} using the
recommended default parameters.

For each combination between a generating scheme and a method, we
average the TPR (True Positive Rate, or sensitivity) and the FPR (False
Positive Rate, one minus specificity) over \(N=200\) repetitions, for
each value of PP (the number of predicted positive) which scales from
\(0\) to \(n\). We also compute point-wise confidence bounds using the
standard deviation of the TPR over the \(N\) repetitions and the
standard Gaussian quantile of order 97.5\%. The ROC curves together with
their confidence bands are represented in
Figure~\ref{fig-comparison_methods}, separately for the three
density-generating processes. Table~\ref{tbl-comparison_auc} summarises
the performance of the methods across the schemes, by means of the
average area under the curve (AUC).

We can see that both ICS methods give quite similar results except for
the \texttt{GP\_clr} generating process where \texttt{ICS\_B2}
outperforms \texttt{ICS\_L2}. Together with \texttt{MUOD\_L2} and
\texttt{MUOD\_QF}, these methods are the best in terms of average AUC,
although ICS-based methods perform more consistently across the
different generating schemes. The \texttt{Median\_LQD} and
\texttt{MBD\_LQD} methods are worse than the others for all generating
schemes. Overall, we can recommend ICS versus the other outlier
detection methods in this situation where the proportion of outliers is
small.

\phantomsection\label{cell-fig-comparison_methods}
\begin{figure}[H]

\centering{

\pandocbounded{\includegraphics[keepaspectratio]{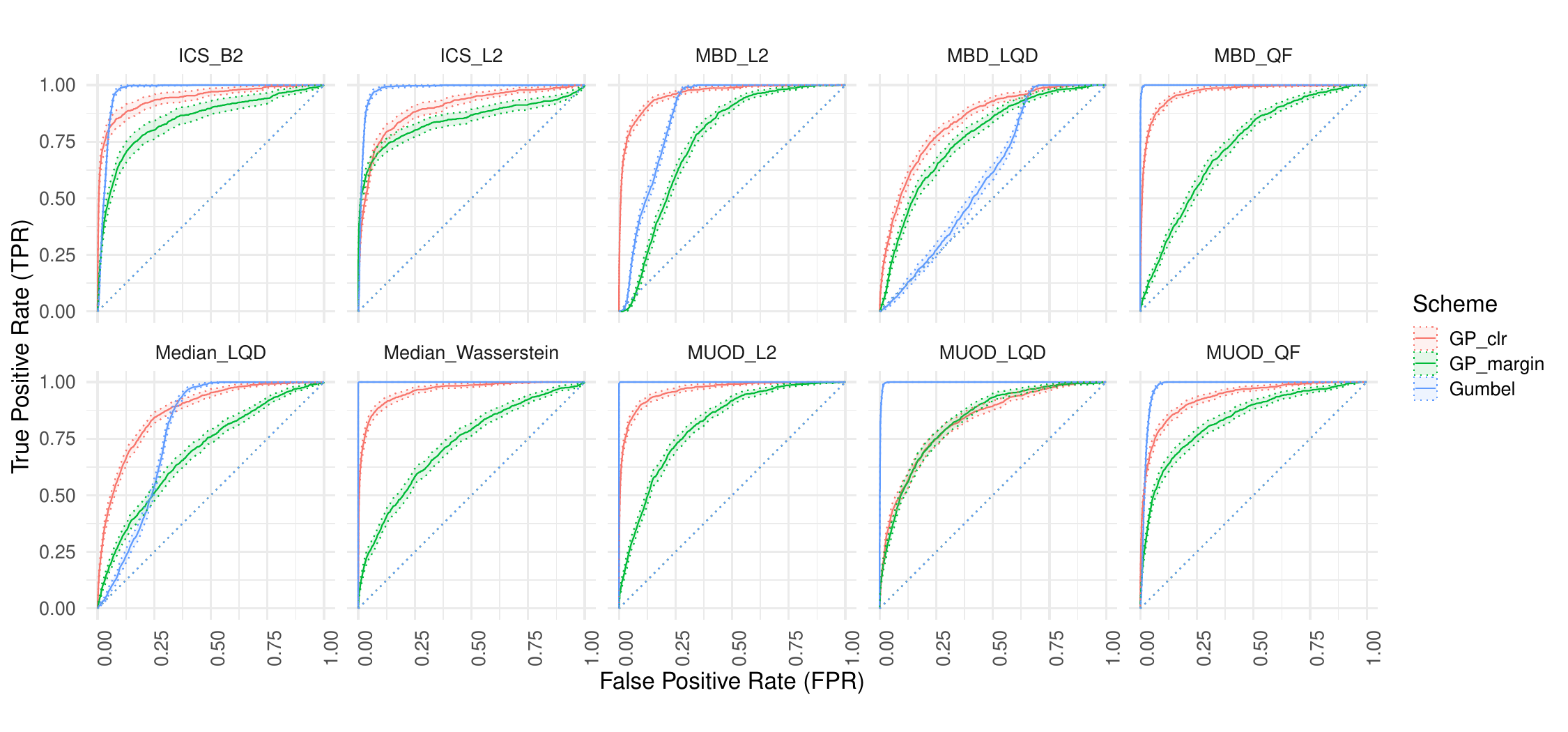}}

}

\caption{\label{fig-comparison_methods}ROC curves of 10 different
outlier detection methods for density data with 3 generating schemes.}

\end{figure}%

\begin{longtable}[]{@{}llr@{}}

\caption{\label{tbl-comparison_auc}AUC for the 10 outlier detection
methods, averaged across the 3 generating schemes.}

\tabularnewline

\toprule\noalign{}
Approach & Metric & Average AUC \\
\midrule\noalign{}
\endhead
\bottomrule\noalign{}
\endlastfoot
MUOD & L2 & 0.92 \\
ICS & B2 & 0.92 \\
MUOD & QF & 0.91 \\
ICS & L2 & 0.91 \\
MBD & QF & 0.90 \\
Median & Wasserstein & 0.90 \\
MUOD & LQD & 0.88 \\
MBD & L2 & 0.86 \\
Median & LQD & 0.78 \\
MBD & LQD & 0.74 \\

\end{longtable}

\section{An application to Vietnamese climate data}\label{sec-appliout}

\subsection{Data description and preprocessing}\label{sec-datades}

\phantomsection\label{cell-fig-vnt_climate_regions_map}
\begin{figure}[H]

\centering{

\pandocbounded{\includegraphics[keepaspectratio]{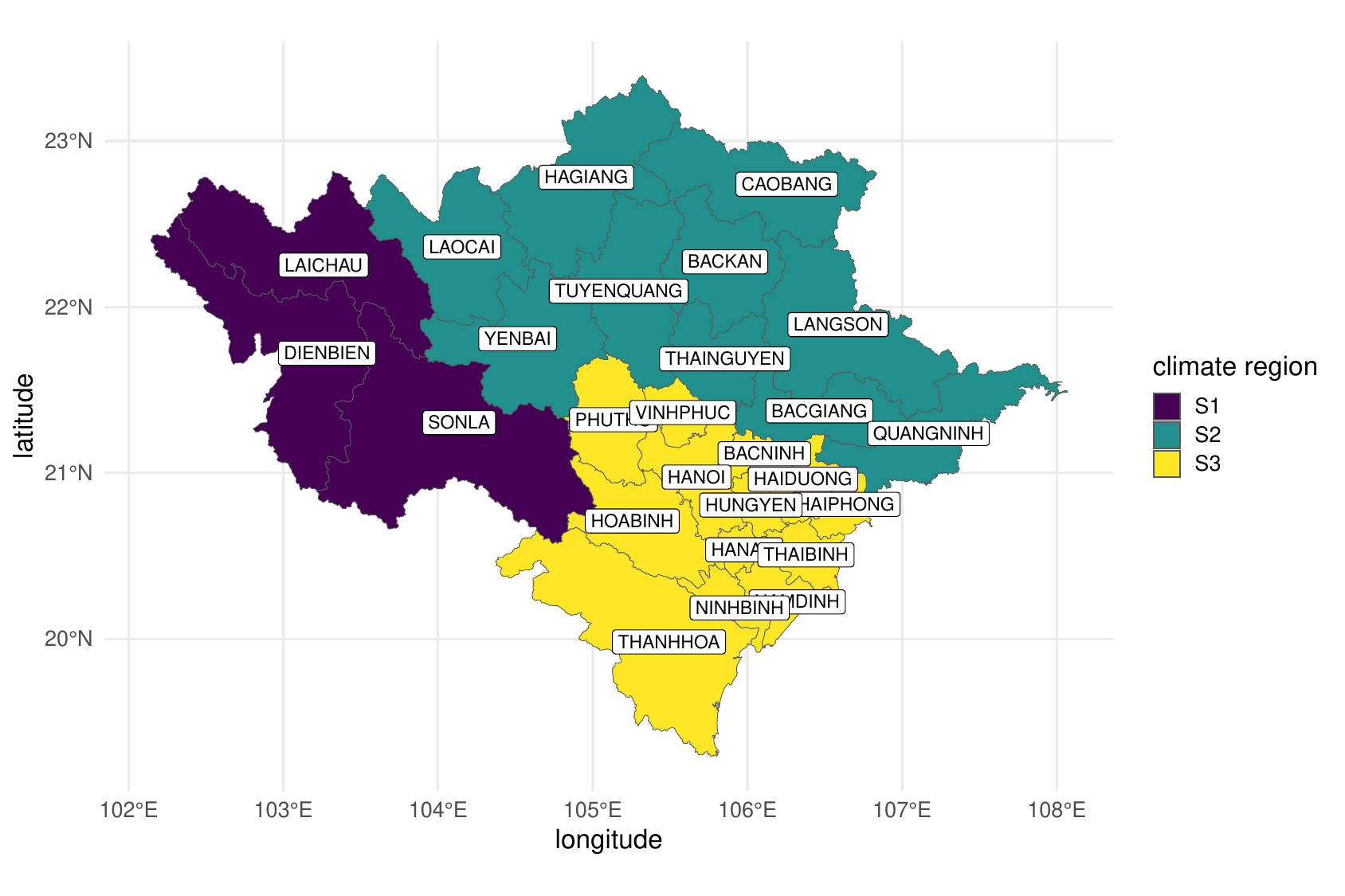}}

}

\caption{\label{fig-vnt_climate_regions_map}The three climate regions of
Northern Vietnam.}

\end{figure}%

In this application, we study daily maximum temperatures for each of the
\(I=63\) Vietnamese provinces over a \(T=30\)-year period (1987-2016).
Originally from the Climate Prediction Center (CPC) database, developed
and maintained by the National Oceanic and Atmospheric Administration
(NOAA), the data underwent a preliminary treatment presented in
\autocite{trinh_discrete_2023}. From the daily 365 or 366 values for
each year, we derive the yearly maximum temperature distribution for
each of the 1,890 province-year units. We assume that the temperature
samples are independent across years and spatially across provinces,
which is a simplifying assumption. Figure~\ref{fig-vnt_north_south_map}
depicts the six administrative regions of Vietnam, and the corresponding
provinces. However, these regions cover areas with varied climates. To
achieve more climatically homogeneous groupings, we use clusters of
provinces based on climatic regions as defined by
\textcite{stojanovic_trends_2020}.
Figure~\ref{fig-vnt_climate_regions_map} displays the three climatic
regions covering Northern Vietnam. We focus on region S3, composed of 13
provinces, by similarity with the North Plain (Red River Delta) (S3) in
\autocite{stojanovic_trends_2020}.

Figure~\ref{fig-vnt_climate_regions_dens} shows the maximum temperature
densities for the 13 provinces of S3, plotted by year, using the
preprocessing detailed in Section~\ref{sec-preproc} with degree less
than or equal to \(d=4\), smoothing parameter \(\lambda=10\) and
\(k=10\) knots located at quantiles of the pooled sample (across space
and time). We observe more variability across time than across space
which confirms that the spatial homogeneity objective is achieved.

\phantomsection\label{cell-fig-vnt_climate_regions_dens}
\begin{figure}[H]

\centering{

\pandocbounded{\includegraphics[keepaspectratio]{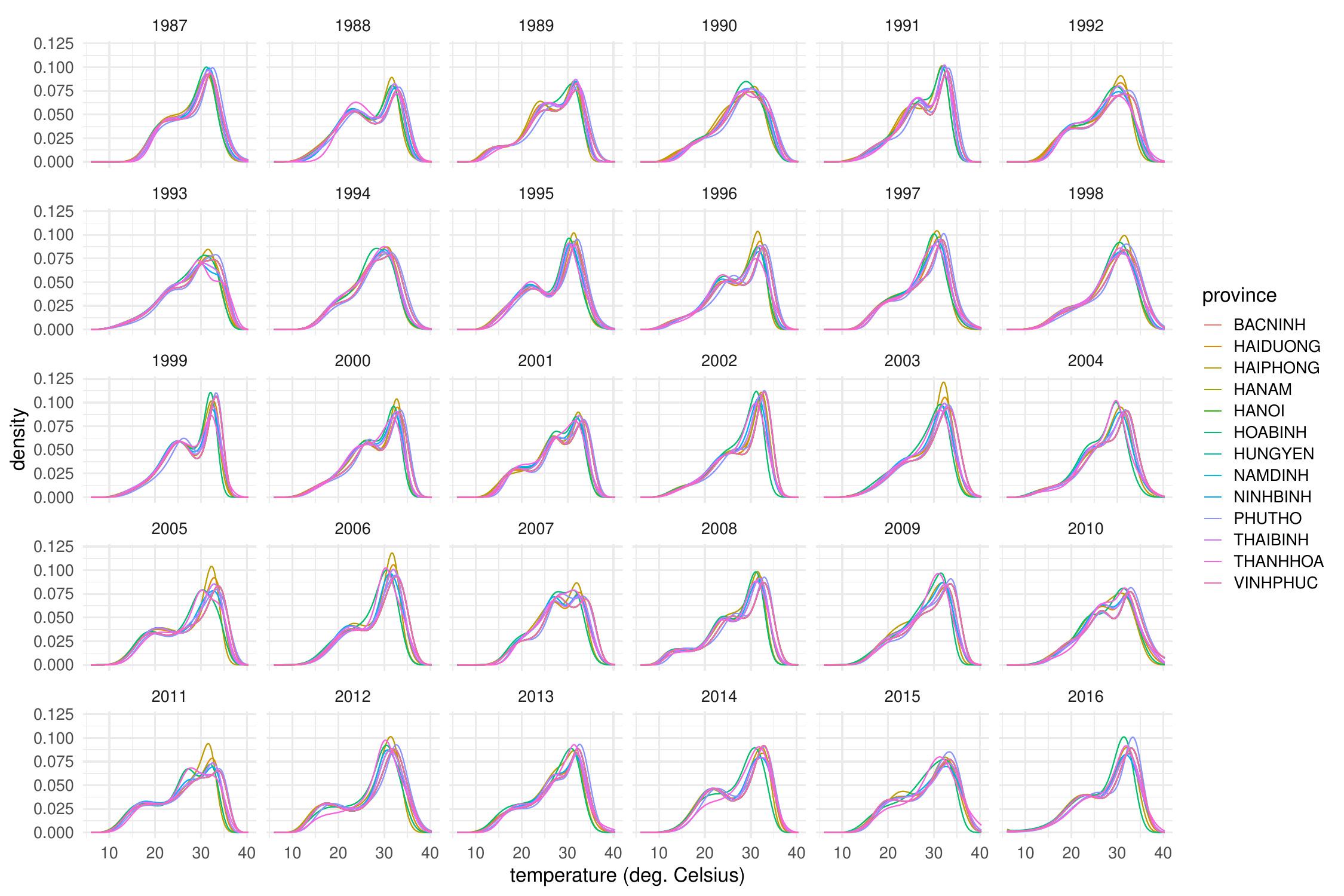}}

}

\caption{\label{fig-vnt_climate_regions_dens}Maximum temperature
densities for the 13 provinces in the S3 climate region of Northern
Vietnam, 1987-2016, colour-coded by province.}

\end{figure}%

\subsection{Outlier detection using ICS for the S3 climate region of
Vietnam}\label{outlier-detection-using-ics-for-the-s3-climate-region-of-vietnam}

We follow the different steps described in Section~\ref{sec-icsout}, and
examine the results of ICS outlier detection using the scatter pair
\((\operatorname{Cov}, \operatorname{Cov}_4)\) on the 390 (13 provinces
\(\times\) 30 years) densities from region S3, obtained after the
preprocessing detailed above.

\begin{figure}

\begin{minipage}{0.50\linewidth}

\centering{

\pandocbounded{\includegraphics[keepaspectratio]{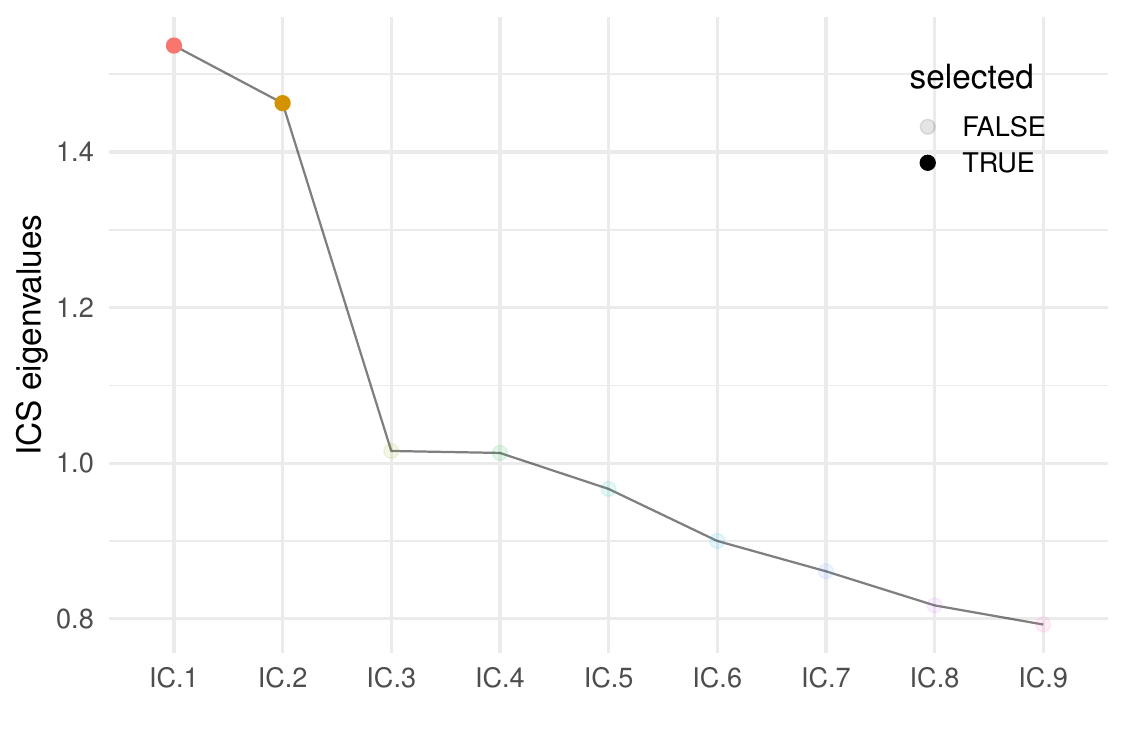}}

}

\subcaption{\label{fig-vnt_climate_regions_screedist-1}}

\end{minipage}%
\begin{minipage}{0.50\linewidth}

\centering{

\pandocbounded{\includegraphics[keepaspectratio]{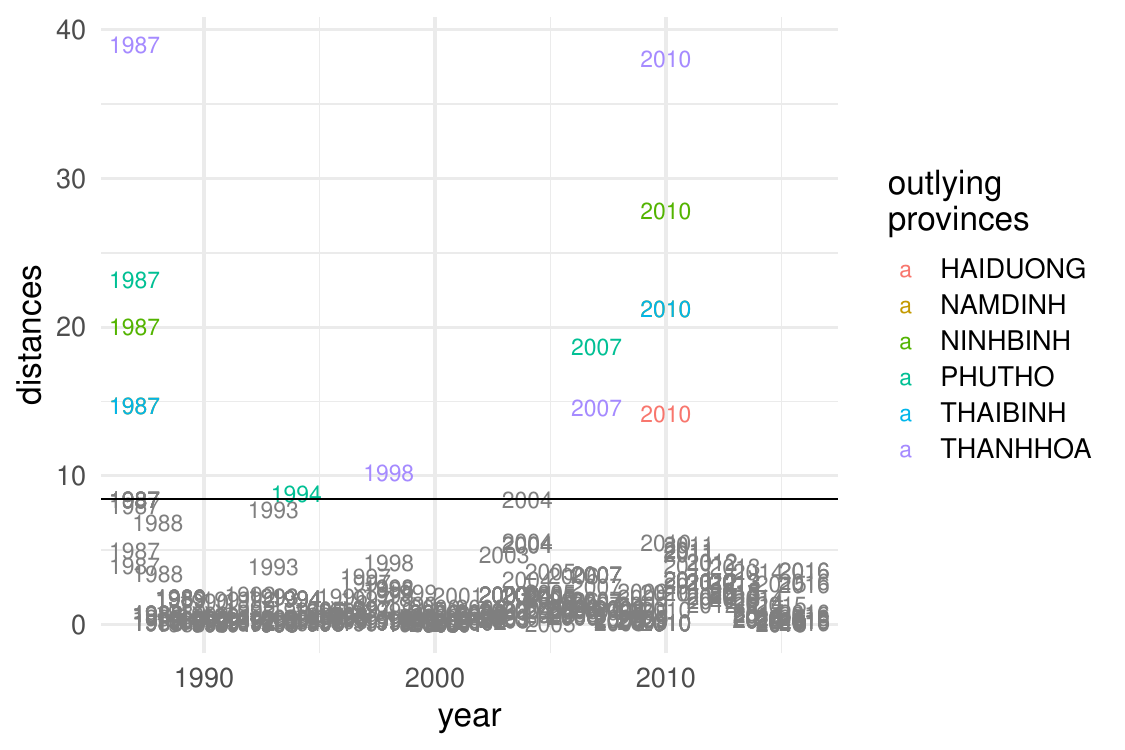}}

}

\subcaption{\label{fig-vnt_climate_regions_screedist-2}}

\end{minipage}%

\caption{\label{fig-vnt_climate_regions_screedist}Scree plot of the ICS
eigenvalues (left panel), and the ICS distances based on the first two
components (right panel) for maximum temperature densities for the 13
provinces in the S3 climate region of Northern Vietnam, 1987-2016.}

\end{figure}%

\begin{figure}

\begin{minipage}{0.50\linewidth}

\centering{

\pandocbounded{\includegraphics[keepaspectratio]{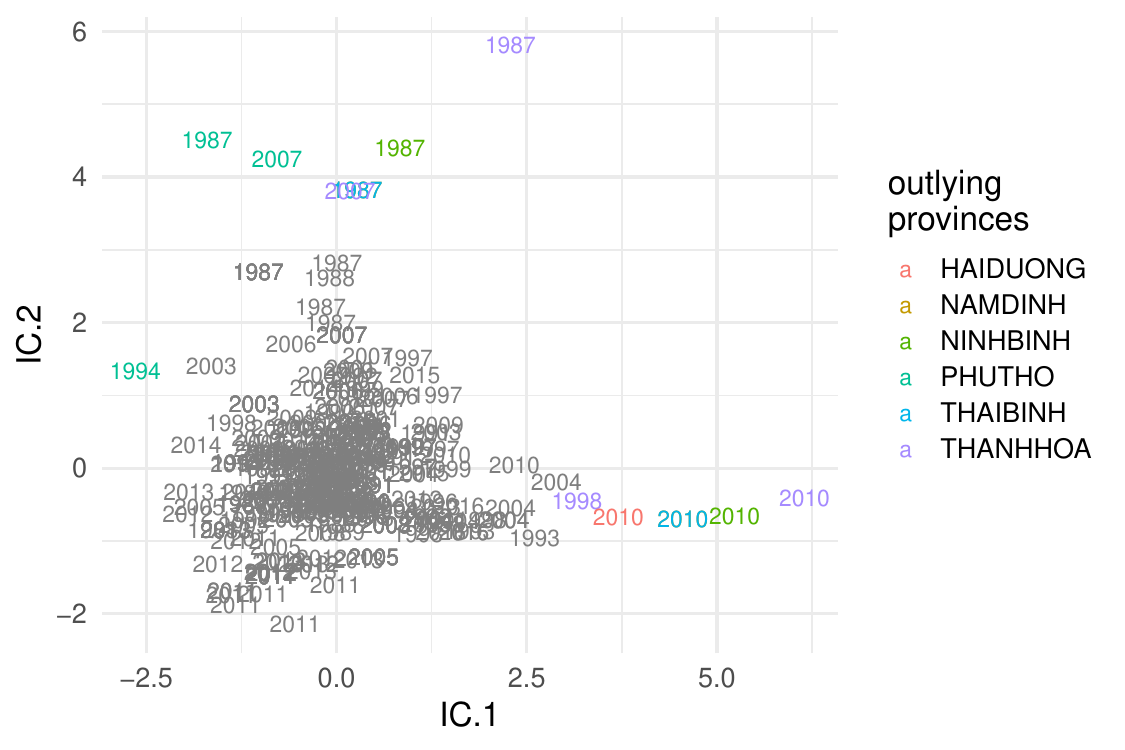}}

}

\subcaption{\label{fig-vnt_climate_regions_eigdensscat-1}}

\end{minipage}%
\begin{minipage}{0.50\linewidth}

\centering{

\pandocbounded{\includegraphics[keepaspectratio]{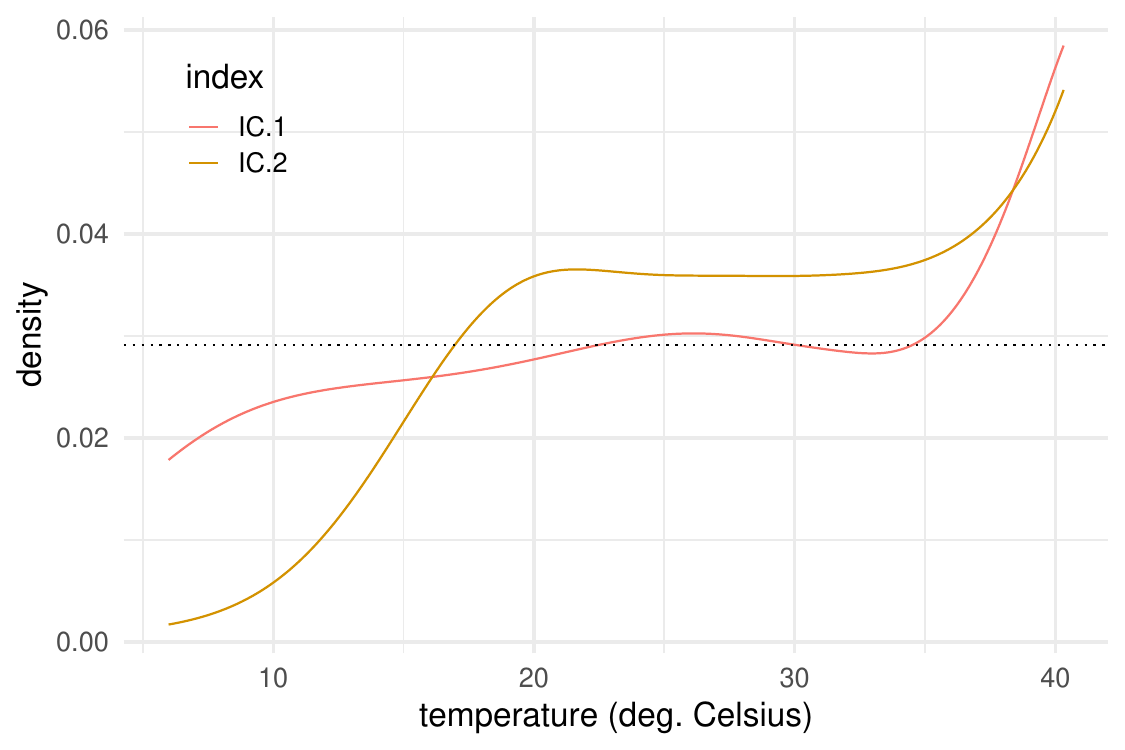}}

}

\subcaption{\label{fig-vnt_climate_regions_eigdensscat-2}}

\end{minipage}%

\caption{\label{fig-vnt_climate_regions_eigdensscat}Scatter plot of the
first two invariant components (left panel) labelled by year and
coloured by province, and the first two ICS dual eigendensities (right
panel) of the maximum temperature densities for the 13 provinces in the
S3 climate region of Northern Vietnam, 1987-2016.}

\end{figure}%

The scree plot on the left panel of
Figure~\ref{fig-vnt_climate_regions_screedist} clearly indicates that we
should retain the first two invariant components. The right panel of
Figure~\ref{fig-vnt_climate_regions_screedist} shows the squared ICS
distances based on these first two components, with the observations
index on the \(x\)-axis and with a threshold (horizontal line)
corresponding to a significance level of \(2.5\%\). This plot reveals
that several observations are distinctly above this threshold,
especially for the years 1987 and 2010.

The left panel of Figure~\ref{fig-vnt_climate_regions_eigdensscat}
displays the scatter plot of the first two components, labelled by year.
The densities are coloured by province for the outliers and coloured in
grey for the other provinces. This plot reveals that the outliers are
either densities from 2010 (and one density from 1998) that are outlying
on the first component, or densities from 1987 and 2007 that are
outlying on the second component.

To interpret the outlyingness, we can use the dual eigendensities
plotted in the right panel of
Figure~\ref{fig-vnt_climate_regions_eigdensscat} together with
Figure~\ref{fig-vnt_climate_regions_densout}, which represents the
densities and their centred log-ratio transformation, colour-coded by
year for the outliers and in grey for the other observations. This is
justified by the reconstruction formula of
Proposition~\ref{prp-reconstruction} in the Appendix. The horizontal
line on the eigendensities plot (right plot of
Figure~\ref{fig-vnt_climate_regions_eigdensscat}) corresponds to the
uniform density on the interval \([5;40]\). Four provinces in 2010 are
outlying with large positive values on the first invariant component
(see the left panel of
Figure~\ref{fig-vnt_climate_regions_eigdensscat}). The first
eigendensity IC.1 is characterised by a smaller mass of the temperature
values on the interval \([5;20]\), compared to the uniform distribution,
a mass similar to the uniform on \([20;35]\), and a much larger mass
than the uniform on the interval \([35;40]\). These four observations
correspond to the four blue curves on the left and right panels of
Figure~\ref{fig-vnt_climate_regions_densout}. Compared to the other
densities, these four densities exhibit relatively lighter tails on the
lower end of the temperature spectrum and heavier tails on the higher
end. For temperature values in the medium range, these four observations
fall in the middle of the cloud of densities and of clr transformed
densities. On the second invariant component, six observations take
large values and are detected as outliers. They correspond to four
provinces in 1987 and three in 2007 (see the left panel of
Figure~\ref{fig-vnt_climate_regions_eigdensscat}). The second
eigendensity IC.2 differs greatly from the uniform distribution on the
whole interval of temperature values. The left tail is much lighter
while the right tail is much heavier. Besides the six observations
flagged as outliers, other provinces in 1987 and 2007 take large values
on IC.2, and correspond to densities with very few days with maximum
temperature less than 15 degrees Celsius compared to other densities.

\begin{figure}

\begin{minipage}{0.50\linewidth}

\centering{

\pandocbounded{\includegraphics[keepaspectratio]{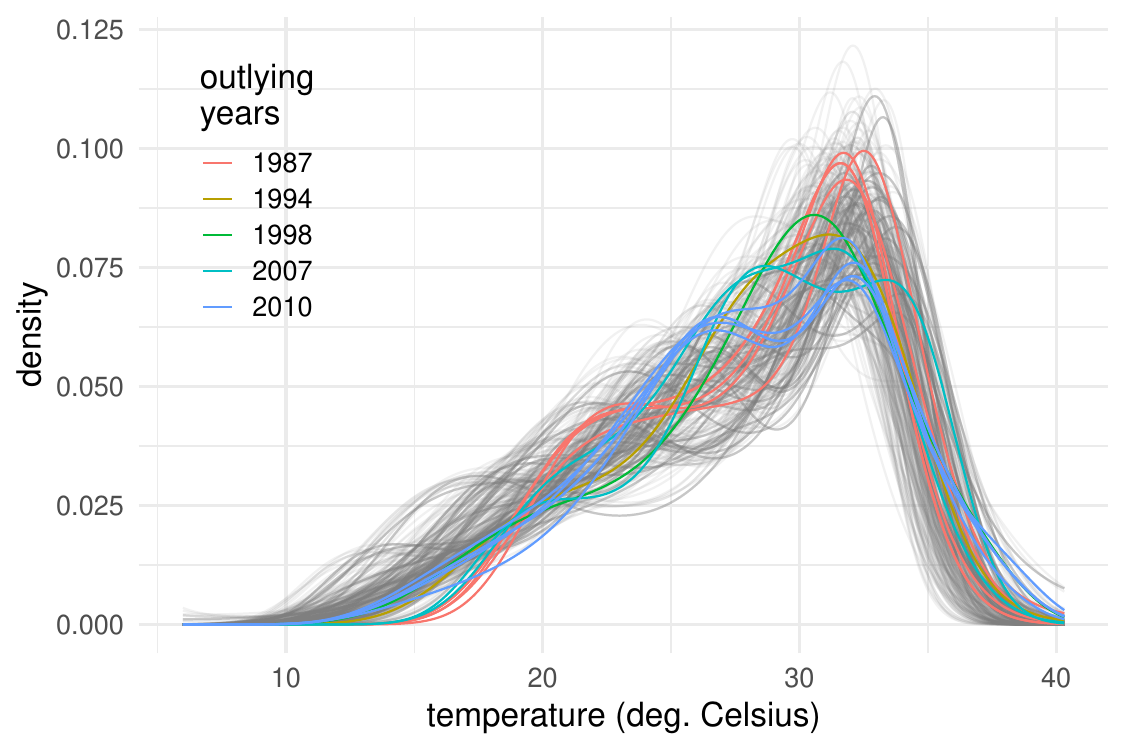}}

}

\subcaption{\label{fig-vnt_climate_regions_densout-1}}

\end{minipage}%
\begin{minipage}{0.50\linewidth}

\centering{

\pandocbounded{\includegraphics[keepaspectratio]{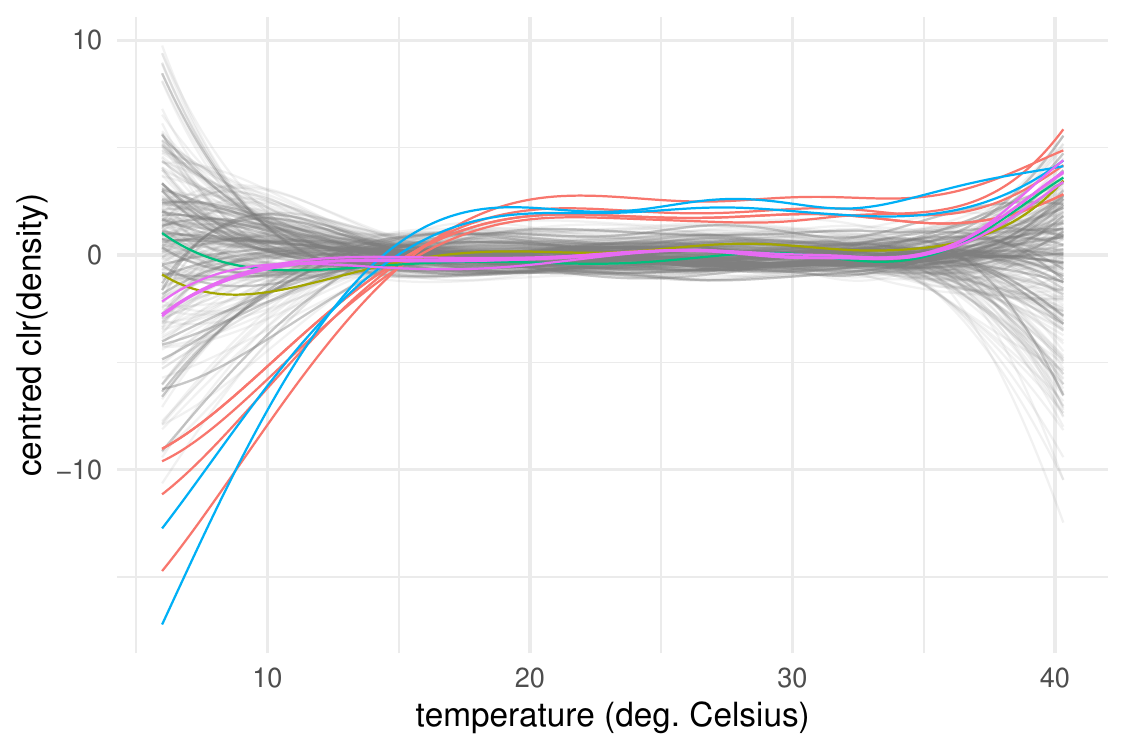}}

}

\subcaption{\label{fig-vnt_climate_regions_densout-2}}

\end{minipage}%

\caption{\label{fig-vnt_climate_regions_densout}Maximum temperature
densities (left panel) and their centred log-ratio transforms (right
panel) for the 13 provinces in the S3 climate region of Northern
Vietnam, 1987-2016, outlying densities are colour-coded by year.}

\end{figure}%

\subsection{Influence of the preprocessing
parameters}\label{sec-app_smooth}

\phantomsection\label{cell-fig-vnt_climate_regions_grid_outliers}
\begin{figure}[H]

\centering{

\pandocbounded{\includegraphics[keepaspectratio]{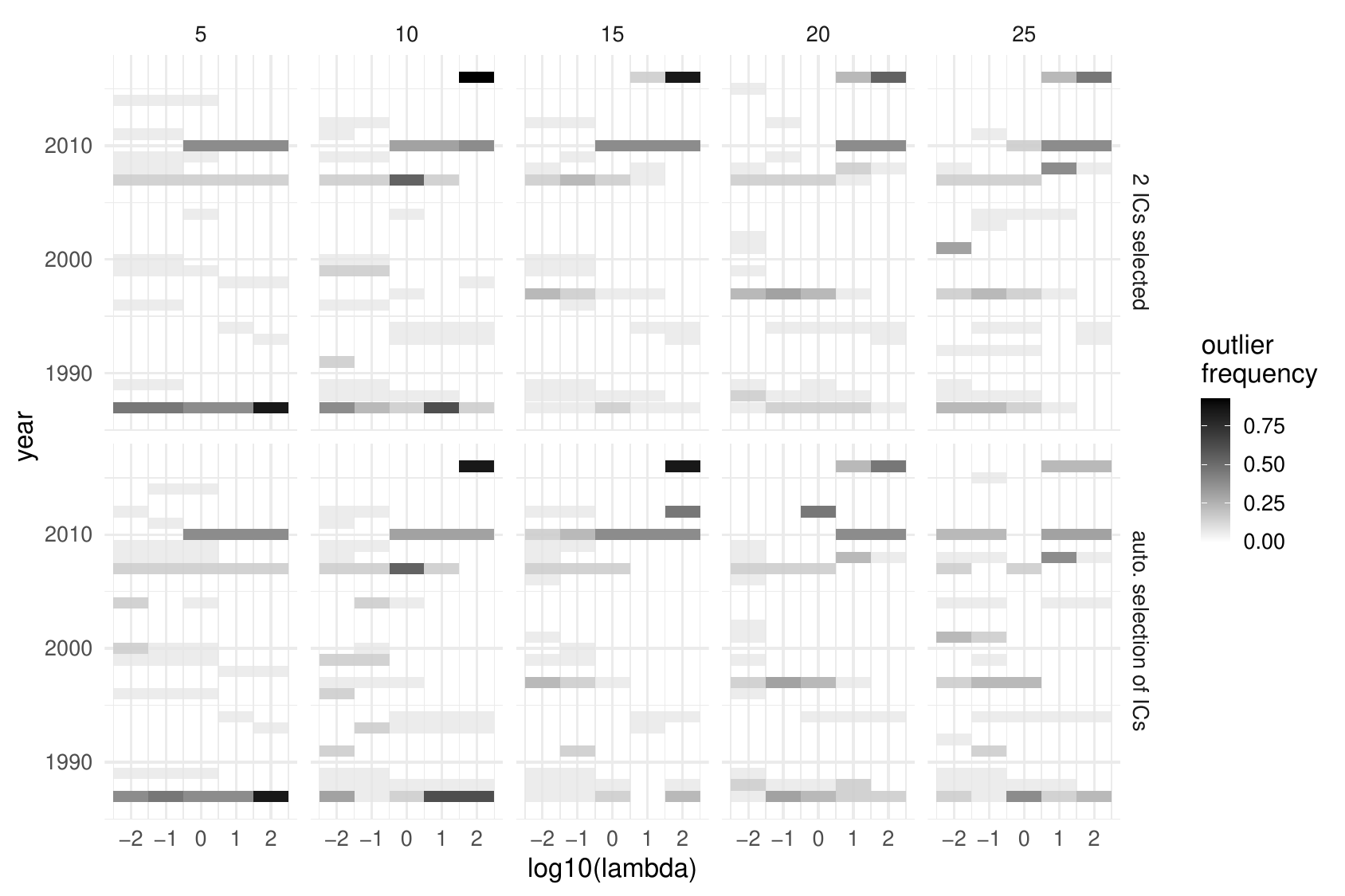}}

}

\caption{\label{fig-vnt_climate_regions_grid_outliers}Outlier detection
by ICS across smoothing parameters for the Vietnam climate data.
\emph{Top:} 2 invariant components selected; \emph{Bottom:} automatic
selection through D'Agostino tests. \emph{\(y\)-axis:} year;
\emph{\(x\)-axis:} \(\lambda\) parameter. Columns correspond to knot
numbers (5-25). Outliers are marked as light gray to black squares
depending on their detection frequency.}

\end{figure}%

\phantomsection\label{cell-fig-vnt_climate_regions_grid_summary}
\begin{figure}[H]

\centering{

\pandocbounded{\includegraphics[keepaspectratio]{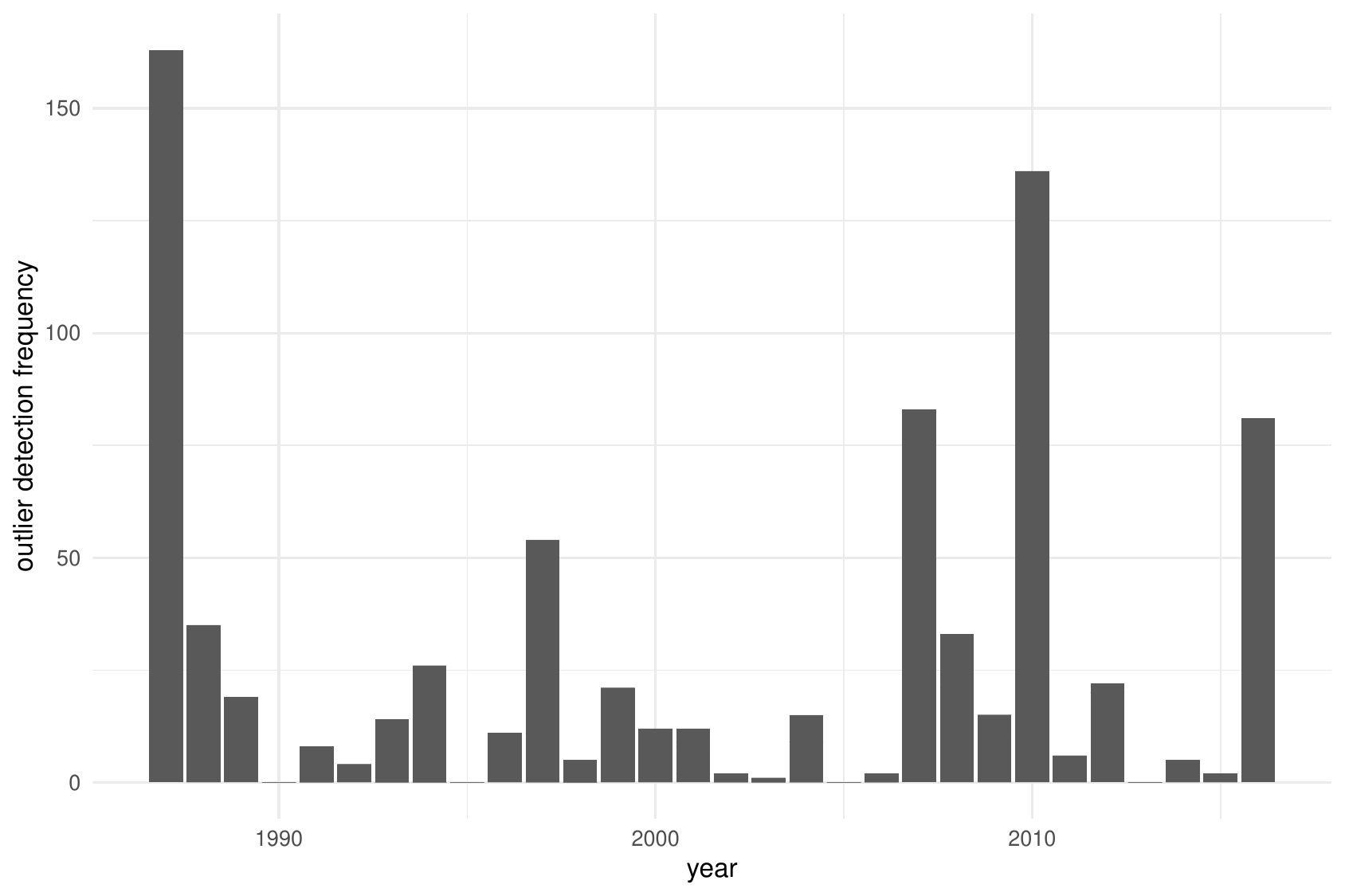}}

}

\caption{\label{fig-vnt_climate_regions_grid_summary}Frequency of
outlier detection by ICS across all 25 scenarios with varying smoothing
parameters and all 13 provinces, for each year in the Vietnamese climate
dataset.}

\end{figure}%

As mentioned in Section~\ref{sec-toyex}, we can validate the atypical
nature of observations by running the ICSOutlier procedure multiple
times with varying smoothing parameter values. Following the rule of
thumb of one dimension per 10 observations, with 390 observations, we
should consider less than 35 interior knots. In what follows, we take 5,
10, 15, 20 and 25 interior knots and we consider base-10 logarithm
values for \(\lambda\) equal to -2, -1, 0, 1 and 2. The number of
selected ICS components is either fixed equal to 2, or is automatically
determined using the D'Agostino normality test described in
Section~\ref{sec-icsout}. We compute the squared ICS distances of the
390 observations, and observations are classified as outliers when their
squared distance exceeds the threshold based on a \(2.5\%\) level as
detailed in Section~\ref{sec-icsout}.

We plot in Figure~\ref{fig-vnt_climate_regions_grid_outliers} the years
on the \(y\)-axes for the 25 smoothing parameter setups, indicating
outlying years with light gray to black squares depending on their
detection frequency. Figure~\ref{fig-vnt_climate_regions_grid_summary}
displays a bar plot of the frequency of outlier detection (across the 25
setups and the 13 provinces) for each year. Note that the choice of the
number of selected invariant components has minimal impact. Both
Figure~\ref{fig-vnt_climate_regions_grid_outliers} and
Figure~\ref{fig-vnt_climate_regions_grid_summary} confirm the results of
the previous section. Most provinces are outlying in 1987 and several
are also outlying in 2007 and 2010. For large values of \(\lambda\),
many provinces are also detected as outliers in 2016. Some provinces are
detected quite often over the years: THANH HOA, HAI PHONG and HOA BINH.
Note that in \autocite{stojanovic_trends_2020}, the province of THANH
HOA extends across two climatic regions (S3 and S4) which could explain
why it is very often detected as an outlier.

An overall comment regarding the outlier detection procedure that we use
in the present application is that, from our experience on other data
sets, an outlying density is often characterised by a behaviour that
differs from the other densities in the tails of the distribution. This
is not surprising because the Bayes inner product defined by equation
Equation~\ref{eq-innerpdt} involves the ratio of densities which can be
large when a density is small (at the tails of the distribution).

\section{Conclusion and perspectives}\label{sec-conclusion}

We propose a coordinate-free presentation of ICS that allows ICS to be
applied to more complex objects than the coordinates vectors of
multivariate analysis. We focus on the case of distributional data and
describe an outlier detection procedure based on ICS. However, one of
the limitations of the coordinate-free approach is that it is mainly
adapted to pairs of weighted covariance operators, because they have a
coordinate-free definition. These pairs of operators include the
well-known \((\operatorname{Cov}, \operatorname{Cov}_4)\) pair. Its
scatter counterpart in the multivariate context is the one recommended
by \textcite{archimbaud_detection_2018} for a small proportion of
outliers. But it is unclear how we could generalise other well-known
scatter matrices (such as M-estimators, pairwise-based weighted
estimators, or Minimum Covariance Determinant estimators) which are
useful when using ICS as a preprocessing step for clustering
\autocite[see][]{alfons_tandem_2024}.

Concerning a further adaptation of ICSOutlier to density objects, one
perspective to our work is to take into account different settings for
the preprocessing parameters and aggregate the results in a single
outlyingness index. Another perspective is to consider multivariate
densities (e.g., not only maximum density temperature but also minimum
density temperature, precipitation,\ldots) and generalise the ICSOutlier
procedure as in \autocite{archimbaud_ics_2022} for multivariate
functional data.

This coordinate-free framework for ICS lays the groundwork for a
generalisation to infinite-dimensional Hilbert spaces. Many difficulties
arise, such as the compactness of the covariance operator which makes it
non surjective, so that one cannot easily define a Mahalanobis distance,
on which our definition of weighted covariance operators relies.
Moreover, the existence of solutions and other properties of ICS proved
in this paper come from the fact that one of the scatter operators is an
automorphism, so it cannot be compact (in particular not the
covariance). Finally, \textcite{tyler_note_2010} proved that, whenever
the dimension \(p\) is larger than the number of observations \(n\), all
affine equivariant scatter operators are proportional, which is a bad
omen for a straight generalisation to infinite-dimensional Hilbert
spaces. One can partially circumvent these difficulties by assuming that
the data is almost surely in a deterministic finite-dimensional subspace
\(E\) of \(H\) (which is the case for density data after our
preprocessing) and applying coordinate-free ICS. Another option could be
to alleviate the affine equivariance assumption.

\section*{Acknowledgments}\label{acknowledgments}
\addcontentsline{toc}{section}{Acknowledgments}

The major part of this work was completed while the authors were
visiting the Vietnam Institute for Advanced Study in Mathematics (VIASM)
in Hanoi and the authors express their gratitude to VIASM. This paper
has also been funded by the Agence Nationale de la Recherche under grant
ANR-17-EURE-0010 (Investissements d'Avenir program). We thank Thibault
Laurent for attracting our attention on the climate regions partition of
Vietnam. We also thank the two reviewers who gave us constructive
comments that allowed us to improve our article.

\section*{Appendix}\label{appendix}
\addcontentsline{toc}{section}{Appendix}

\subsection*{Scatter operators for random objects in a Hilbert
space}\label{scatter-operators-for-random-objects-in-a-hilbert-space}
\addcontentsline{toc}{subsection}{Scatter operators for random objects
in a Hilbert space}

Let us first discuss some definitions relative to scatter operators in
the framework of a Hilbert space \((E, \langle \cdot, \cdot \rangle)\).
We consider an \(E\)-valued random object \(X: \Omega \rightarrow E\)
where \(\Omega\) is a probability space and \(E\) is a Hilbert space
equipped with the Borel \(\sigma\)-algebra. In order to define ICS, we
need at least two scatter operators, which generalise the covariance
operator defined on \(E\) by
\begin{equation}\phantomsection\label{eq-cov}{\forall (x,y) \in E^2, \langle \operatorname{Cov} [X] x, y \rangle = \mathbb E \left[ \langle X - \mathbb EX, x \rangle \langle X - \mathbb EX, y \rangle \right],}\end{equation}
while keeping its affine equivariance property:
\[\forall A \in \mathcal{GL} (E), \forall b \in E, \operatorname{Cov} [AX+b] = A \operatorname{Cov} [X] A^*,\]
where the Hilbert norm of \(X\) is assumed to be square-integrable, and
\(A^*\) is the adjoint linear operator of \(A\) in the Hilbert space
\(E\), represented by the transpose of the matrix that represents \(A\).

\begin{definition}[Scatter
operators]\protect\hypertarget{def-scatter}{}\label{def-scatter}

Let \((E, \langle \cdot, \cdot \rangle)\) be a Hilbert space of
dimension \(p\), \(\mathcal E\) an affine invariant set of \(E\)-valued
random objects, i.e.~that verifies:
\begin{equation}\phantomsection\label{eq-affine-inv-set}{\forall X \in \mathcal E, \forall A \in \mathcal{GL} (E), \forall b \in E, AX+b \in \mathcal E.}\end{equation}
An operator \(S: \mathcal E \rightarrow \mathcal S^+ (E)\) (where
\(\mathcal S^+ (E)\) is the space of non-negative symmetric operators on
\(E\)) is called an (affine equivariant) scatter operator (defined on
\(\mathcal E\)) if it satisfies the following two properties:

\begin{enumerate}
\def\labelenumi{\arabic{enumi}.}
\item
  Invariance by equality in distribution:
  \[\forall (X,Y) \in \mathcal E^2, X \sim Y \Rightarrow S[X] = S[Y].\]
\item
  Affine equivariance:
  \[\forall X \in \mathcal E, \forall A \in \mathcal{GL} (E), \forall b \in E, S[AX+b] = A S[X] A^*.\]
\end{enumerate}

\end{definition}

We do not know whether there exist other scatter operators than the
covariance when the Hilbert space has infinite dimension.

\subsection*{Details on coordinate-free
ICS}\label{details-on-coordinate-free-ics}
\addcontentsline{toc}{subsection}{Details on coordinate-free ICS}

The problem \(\operatorname{ICS} (X, S_1, S_2)\) defined by
Equation~\ref{eq-icsdef} is equivalent to assuming that \(S_1 [X]\) is
injective and finding an orthonormal basis \(H\) that diagonalises the
non-negative symmetric operator \(S_1 [X]^{-1} S_2 [X]\) in the
Euclidean space \((E, \langle S_1[X] \cdot, \cdot \rangle)\). The
\(\operatorname{ICS} (X,{S_1},{S_2})\) spectrum \(\Lambda\) is unique
and is simply the spectrum of \(S_1 [X]^{-1} S_2 [X]\).

\begin{proposition}[Existence of
solutions]\protect\hypertarget{prp-existence}{}\label{prp-existence}

Let \((E, \langle \cdot, \cdot \rangle)\) be a Euclidean space of
dimension \(p\), \(\mathcal E \subseteq L^1 (\Omega, E)\) an affine
invariant set of integrable \(E\)-valued random objects, \(S_1\) and
\(S_2\) two scatter operators on \(\mathcal E\). For any
\(X \in \mathcal E\) such that \(S_1 [X]\) is an automorphism, there
exists at least one solution \((H, \Lambda)\) to the problem
\(\operatorname{ICS}(X, S_1, S_2)\), and \(\Lambda\) is a uniquely
determined non-increasing sequence of positive real numbers.

\end{proposition}

\begin{proof}
Since \(S_1 [X]\) is non-singular, \(S_1 [X]^{-1} S_2 [X]\) exists and
is symmetric in the Euclidean space
\((E, \langle {S_1}[X] \cdot, \cdot \rangle)\), because
\[\begin{aligned}
\forall (x,y) \in E^2, \langle S_1[X]  S_1 [X]^{-1} S_2 [X] x, y \rangle &= \langle S_2 [X]x, y \rangle = \langle S_2 [X] y, x \rangle \\
&= \langle S_1[X] S_1 [X]^{-1} S_2 [X] y, x \rangle.
\end{aligned}\] Thus, the spectral theorem guarantees that there exists
an orthonormal basis \(H\) of
\((E, \langle {S_1}[X] \cdot, \cdot \rangle)\) in which
\(S_1 [X]^{-1} S_2[X]\) is diagonal.~
\end{proof}

This methodology does not generalise to the infinite-dimensional case,
because the inner product space
\((\mathcal H, \langle \cdot, S_1 [X] \cdot \rangle)\) is not
necessarily complete, so the spectral theorem does not apply.

\begin{remark}[Courant-Fischer variational principle]
The ICS problem Equation~\ref{eq-icsdef} can be stated as a maximisation
problem. If \(1 \leq j \leq p\), the following equalities hold:
\begin{equation}\phantomsection\label{eq-Courant}{h_j \in \underset{h \in E, \langle S_1 [X] h, h_{j'} \rangle = 0 \text{ if } 0 < j' < j}{\operatorname{argmax}} \frac{\langle{S_2} [X] h, h \rangle}{\langle{S_1} [X] h, h \rangle} \;\text{ and }\;
\lambda_j = \max_{h \in E, \langle S_1 [X] h, h_{j'} \rangle = 0 \text { if } 0 < j' < j} \frac{\langle{S_2} [X] h, h \rangle}{\langle{S_1} [X] h, h \rangle}.}\end{equation}
\end{remark}

The following reconstruction formula, extended from multivariate to
complex data, is useful to interpret the ICS dual eigenbasis
\(H^* = (h^*_j)_{1 \leq j \leq p}\), which is defined as the only basis
of the space \(E\) that satisfies
\[\langle h_j, h^*_{j'} \rangle = \delta_{jj'} \text{ for all }1 \leq j, j' \leq p.\]

\begin{proposition}[Reconstruction
formula]\protect\hypertarget{prp-reconstruction}{}\label{prp-reconstruction}

Let \((E, \langle \cdot, \cdot \rangle)\) be a Euclidean space of
dimension \(p\), \(\mathcal E \subseteq L^1 (\Omega, E)\) an affine
invariant set of integrable \(E\)-valued random objects, \(S_1\) and
\(S_2\) two scatter operators on \(\mathcal E\). For any
\(X \in \mathcal E\) such that \(S_1 [X]\) is an automorphism and any
\(\operatorname{ICS}(X, S_1, S_2)\) eigenbasis
\({H} = (h_1, \dots, h_p)\) of \(E\), we have
\[X = \mathbb EX + \sum_{j=1}^p z_j h^*_j,\] where the
\(z_j, 1 \leq j \leq p\) are defined as in Equation~\ref{eq-ic} and
\(H^* = (h^*_j)_{1 \leq j \leq p} = ({S_1} [X] h_j)_{1 \leq j \leq p}\)
is the dual basis of \(H\).

\end{proposition}

\subsection*{Reminder on Bayes spaces}\label{reminder-on-bayes-spaces}
\addcontentsline{toc}{subsection}{Reminder on Bayes spaces}

The most recent and complete description of the Bayes spaces approach
can be found in \autocite{van_den_boogaart_bayes_2014}. For the present
work, we will identify the elements of a Bayes space, as defined by
\textcite{van_den_boogaart_bayes_2014}, with their Radon--Nikodym
derivative with respect to a reference measure \(\lambda\). This leads
to the following framework: let \((a,b)\) be a given interval of the
real line equipped with the Borel \(\sigma\)-algebra, let \(\lambda\) be
a finite reference measure on \((a,b)\). Let \(B^2(a,b)\) be the space
of square-log integrable probability densities
\(\frac{d\mu}{d\lambda},\) where \(\mu\) is a measure that is equivalent
to \(\lambda\), which means that \(\mu\) and \(\lambda\) are absolutely
continuous with respect to each other.

Note that the simplex \(\mathcal S^p\) used in compositional data
analysis can be seen as a Bayes space when considering, instead of an
interval \((a,b)\) equipped with the Lebesgue measure, the finite set
\(\{1, \dots, p+1\}\) equipped with the counting measure \autocite[see
Example 2 in][]{van_den_boogaart_bayes_2014}.

Let us first briefly recall the construction of the Hilbert space
structure of \(B^2(a,b)\). For a density \(f\) in \(B^2(a,b)\), the clr
transformation is defined by
\[\operatorname{clr} f(.) = \log f(.) - \frac{1}{\lambda(a,b)}\int_{a}^{b}\log
    f(t)d\lambda(t).\] The clr transformation maps an element of
\(B^2(a,b)\) into an element of the space \(L^2_0(a,b)\) of functions
which are square-integrable with respect to \(\lambda\) on \((a,b)\) and
whose integral is equal to zero. The clr inverse of a function \(u\) of
\(L^2_0(a,b)\) is \({B}^2\)-equivalent to \(\exp(u).\) More precisely,
if \(u \in L^2_0 (a,b),\)
\[\operatorname{clr}^{-1}(u)(.) = \frac{\exp u(.)}{\int_a^b \exp{u(t)  d\lambda(t)}}.\]
A vector space structure on \(B^2(a,b)\) is readily obtained by
transporting the vector space structure of \(L^2_0 (a,b)\) to
\(B^2(a,b)\) using the clr transformation and its inverse, see for
example \textcite{van_den_boogaart_bayes_2014}. Its operations, denoted
by \(\oplus\) and \(\odot\), are called perturbation (the ``addition'')
and powering (the ``scalar multiplication'').

For the definition of the inner product, we adopt a normalization
different from that of \textcite{egozcue_hilbert_2006} and of
\textcite{van_den_boogaart_bayes_2014} in the sense that we choose the
classical definition of inner product in \(L^2_0(a,b),\) for two
functions \(u\) and \(v\) in \(L^2_0 (a,b)\)
\begin{equation}\phantomsection\label{eq-inner}{\langle u,v \rangle_{L^2_0} = \int_a^b u(t)v(t)d\lambda(t),}\end{equation}
so that the corresponding inner product between two densities \(f\) and
\(g\) in the Bayes space \(B^2(a,b)\) is given by
\begin{equation}\phantomsection\label{eq-innerpdt}{\langle f ,g \rangle_{B^2} = \frac{1}{2 \lambda(a,b)}\int_a^b \int_a^b 
    (\log f(t) - \log f(s))
    (\log g(t) - \log g(s))
    d \lambda(t)d \lambda(s).}\end{equation} This normalization yields
an inner product which is homogeneous to the measure \(\lambda\) whereas
the \textcite{van_den_boogaart_bayes_2014} normalization is unitless.
Note that, for clarity and improved readability, the interval over which
the spaces \(L^2_0\) and \(B^2\) are defined are omitted from some
notations.

For a random density \(f(.)\) in the infinite-dimensional space
\({B}^2(a,b)\), the expectation and covariance operators can be defined
as follows, whenever they exist:
\begin{equation}\phantomsection\label{eq-covw2}{\begin{aligned}
\mathbb{E}^{B^2} [ f ] &= \operatorname{clr}^{-1} \mathbb{E} [ \operatorname{clr} f] \in {B}^2 (a,b) \\
\operatorname{Cov}^{B^2} [f] g &= \mathbb{E}^{B^2} \left[ \langle f \ominus \mathbb{E}^{B^2} [ f ], g \rangle_{B^2} \odot ( f \ominus \mathbb{E}^{B^2} [ f ]) \right] \nonumber \\
&= \operatorname{clr}^{-1}\mathbb{E} [ \langle f, g \rangle_{B^2} \operatorname{clr}f ]  \nonumber \\
&= \operatorname{clr}^{-1}\mathbb{E} [ \langle \operatorname{clr}f, \operatorname{clr} g \rangle_{{ L}^2_0} \operatorname{clr}f ] \quad \text{ for any } g \in B^2(a,b),
\end{aligned}}\end{equation} where \(\ominus\) is the negative
perturbation defined by \(f \ominus g = f \oplus [(-1) \odot g]\).

\subsection*{Reminder on compositional
splines}\label{reminder-on-compositional-splines}
\addcontentsline{toc}{subsection}{Reminder on compositional splines}

Following \autocite{machalova_compositional_2021}, in order to construct
a basis of \(E = \mathcal C^{\Delta \gamma}_d (a,b)\), which is required
in practice, it is convenient to first construct a basis of a
finite-dimensional spline subspace of \(L^2_0(a,b)\), which we then
transfer to \(B^2(a,b)\) by the inverse \(\operatorname{clr}\)
transformation. More precisely, \textcite{machalova_preprocessing_2016}
propose a basis of zero-integral splines in \(L^2_0(a,b)\) that are
called ZB-splines. The corresponding inverse images of these basis
functions by clr are called CB-splines.

A ZB-spline basis, denoted by \(Z = \{Z_1, \ldots, Z_{k+d-1}\},\) is
characterised by the spline of degree less than or equal to \(d\) (order
\(d+1\)), the number \(k\) and the positions of the so-called inside
knots \(\Delta \gamma = \{ \gamma_1, \ldots, \gamma_d\}\) in \((a,b)\).
The dimension of the resulting subspace \({\cal Z}_d^{\Delta \gamma}\)
is \(p=k+d\). Let \({\cal C}_d^{\Delta \gamma}\) be the subspace
generated by \(C=\{C_1, \ldots, C_p\}\) in \(B^2(a,b)\), where
\(C_j = \operatorname{clr}^{-1}(Z_j)\) are the back-transforms in
\(B^2(a,b)\) of the basis functions of the subspace
\({\cal Z}_d^{\Delta \gamma}\). The expansion of a density \(f\) in
\(B^2(a,b)\) is then given by
\begin{equation}\phantomsection\label{eq-expan_B2}{f(t) = \bigoplus_{j=1}^p [f]_{C_j} C_j(t),}\end{equation}
so that the corresponding expansion of its \(\operatorname{clr}\) in
\(L^2_0(a,b)\) is given by
\begin{equation}\phantomsection\label{eq-expan_L2}{\operatorname{clr} f(t) = \sum_{j=1}^p [f]_{C_j} Z_j(t).}\end{equation}

Note that the coordinates of \(f\) in the basis \(C\) are the same as
the coordinates of \(\operatorname{clr}(f)\) in the basis \(Z\), for
\(j=1, \ldots, p, [f]_{C_j} = [\operatorname{clr} f]_{Z_j}.\) Following
\textcite{machalova_preprocessing_2016}, the basis functions of
\({\cal Z}_d^{\Delta \gamma}\) can be written in a B-spline basis, see
\textcite{schumaker_spline_1981}, which is convenient to allow using
existing code for their computation.

\subsection*{Proofs}\label{proofs}
\addcontentsline{toc}{subsection}{Proofs}

\begin{proof}[Proposition~\ref{prp-isometry}]
First, let us verify that the problem
\(\operatorname{ICS} (X^{\mathcal F}, S_1^{\mathcal F}, S_2^{\mathcal F})\)
is well defined on \(F\):

\begin{enumerate}
\def\labelenumi{(\alph{enumi})}
\item
  The application \(\varphi\) is linear so it is measurable. Moreover,
  if \(X \in \mathcal E\), \(A \in \mathcal{GL} (F)\) and \(b \in F\),
  then \[\|\varphi(X)\|_F = \| X \|_E\] and
  \[A \varphi(X) + b = \varphi \left( \varphi^{-1} \circ A \circ \varphi (X) + \varphi^{-1}(b) \right) \text{ where } \varphi^{-1} \circ A \circ \varphi (X) + \varphi^{-1}(b) \in \mathcal E.\]
\item
  If \(X \in \mathcal E\),
  \(S_\ell^{\mathcal F} [\varphi (X)] = \varphi \circ S_\ell^{\mathcal E} [X] \circ \varphi^{-1}\)
  is a non-negative symmetric operator and if \(Y \in \mathcal E\)
  verifies \(\varphi (X) \sim \varphi (Y)\), then \(X \sim Y\) (because
  the Borel \(\sigma\)-algebra on \(E\) is the pullback by \(\varphi\)
  of that on \(F\)) so that, for \(\ell \in \{ 1, 2 \}\),
  \[S_\ell^{\mathcal F} [\varphi (X)] = \varphi \circ S_\ell^{\mathcal E} [X] \circ \varphi^{-1} = \varphi \circ S_\ell^{\mathcal E} [Y] \circ \varphi^{-1} = S_\ell^{\mathcal F} [\varphi (Y)]\]
  and \[\begin{gathered}
          S_\ell^{\mathcal F} [A \varphi(X) + b]
          = \varphi \circ S_\ell^{\mathcal E} [\varphi^{-1} \circ A \circ \varphi (X) + \varphi^{-1}(b)] \circ \varphi^{-1} \\
          = A \circ \varphi \circ S_\ell^{\mathcal E} [X] \circ \varphi^{-1} \circ A^*
          = A S_\ell^{\mathcal F} [\varphi(X)] A^*.
  \end{gathered}\]
\item
  The isometry \(\varphi\) preserves the linear rank of any finite
  sequence of vectors of \(E\).
\end{enumerate}

Now, \((H^{\mathcal E}, \Lambda)\) solves
\(\operatorname{ICS} (X^{\mathcal E}, S_1^{\mathcal E}, S_2^{\mathcal E})\)
in the space \(E\) if and only if \[\begin{aligned}
    &\left\{
    \begin{aligned}
        \langle S_1^{\mathcal E} [X] h_j^{\mathcal E}, h_{j'}^{\mathcal E} \rangle_E &= \delta_{jj'} \text{ for all } 1 \leq j, j' \leq p  \\
        \langle S_2^{\mathcal E} [X] h_j^{\mathcal E}, h_{j'}^{\mathcal E} \rangle_E &= \lambda_j \delta_{jj'} \text{ for all } 1 \leq j, j' \leq p 
    \end{aligned}
    \right. \\
    \iff &\left\{
    \begin{aligned}
        \langle \varphi (S_1^{\mathcal E} [X] h_j^{\mathcal E}), \varphi (h_{j'}^{\mathcal E}) \rangle_F &= \delta_{jj'} \text{ for all } 1 \leq j, j' \leq p  \\
        \langle \varphi (S_2^{\mathcal E} [X] h_j^{\mathcal E}), \varphi (h_{j'}^{\mathcal E}) \rangle_F &= \lambda_j \delta_{jj'} \text{ for all } 1 \leq j, j' \leq p 
    \end{aligned}
    \right. \\
    \iff &\left\{
    \begin{aligned}
        \langle S_1^{\mathcal F} [X] h_j^{\mathcal F}, h_{j'}^{\mathcal F} \rangle_F &= \delta_{jj'} \text{ for all } 1 \leq j, j' \leq p  \\
        \langle S_2^{\mathcal F} [X] h_j^{\mathcal F}, h_{j'}^{\mathcal F} \rangle_F &= \lambda_j \delta_{jj'} \text{ for all } 1 \leq j, j' \leq p,
    \end{aligned}
    \right.
\end{aligned}\] which is equivalent to the fact that
\((H^{\mathcal F}, \Lambda)\) solves
\(\operatorname{ICS} (X^{\mathcal F}, S_1^{\mathcal F}, S_2^{\mathcal F})\)
in the space \(F\).~
\end{proof}

\begin{proof}[Corollary~\ref{cor-isometry-covw}]
Let \(\ell \in \{ 1, 2 \}\) and \(\tilde X = X - \mathbb EX\). In order
to prove the equation Equation~\ref{eq-isometry-covw}, we will need to
prove that, for any \((x,y) \in F^2\),
\begin{equation}\phantomsection\label{eq-cor-isometry-covw}{\begin{aligned}
    \langle \varphi \circ \operatorname{Cov}_{w_\ell}^E [X] \circ \varphi^{-1} (x), y \rangle_F
    &= \langle \operatorname{Cov}_{w_\ell}^E [X] \varphi^{-1} (x), \varphi^{-1} (y) \rangle_E \\
    &= \mathbb E [ w_\ell (\| \operatorname{Cov}^E [X]^{-1/2} \tilde X \|_E)^2 \langle \tilde X, \varphi^{-1} (x) \rangle_E \langle \tilde X, \varphi^{-1} (y) \rangle_E ] \\
    &= \mathbb E [ w_\ell (\| \operatorname{Cov}^F [\varphi(X)]^{-1/2} \varphi(\tilde X) \|_F)^2 \langle \varphi(\tilde X), x \rangle_F \langle \varphi(\tilde X), y \rangle_F ]  \\
    \langle \varphi \circ \operatorname{Cov}_{w_\ell}^E [X] \circ \varphi^{-1} (x), y \rangle_F
    &= \langle \operatorname{Cov}_{w_\ell}^F [\varphi(X)] x, y \rangle_F.
\end{aligned}}\end{equation} It is enough to show the equality between
Equation~\ref{eq-cor-isometry-covw} (2) and
Equation~\ref{eq-cor-isometry-covw} (3), for which we treat differently
the cases \(w_\ell = 1\) and \(w_\ell \neq 1\). If \(w_\ell = 1\), there
is nothing to prove, so that the equation
Equation~\ref{eq-isometry-covw} holds for the covariance operator. If
\(w_\ell \neq 1\), we now know from the case \(w_\ell = 1\) that
\[\operatorname{Cov}^F [\varphi(X)]^{-1/2} = \varphi \circ \operatorname{Cov}^E [X]^{-1/2} \circ \varphi^{-1}\]
so that
\begin{equation}\phantomsection\label{eq-isometry-mahalanobis}{\| \operatorname{Cov}^E [X]^{-1/2} \tilde X \|_E
    = \| \operatorname{Cov}^F [\varphi(X)]^{-1/2} \varphi(\tilde X) \|_F}\end{equation}
Once the equation Equation~\ref{eq-isometry-covw} is proved, one only
needs to apply Proposition~\ref{prp-isometry} to finish the proof.~
\end{proof}

\begin{proof}[Corollary~\ref{cor-coord}]
Applying Corollary~\ref{cor-isometry-covw} to the isometry
\[\varphi_B: \left\{\begin{array}{ccc}
        (E, \langle \cdot, \cdot \rangle_E) &\rightarrow& (\mathbb R^p, \langle \cdot, \cdot \rangle_{\mathbb R^p}) \\
        x &\mapsto& G_B^{1/2} [x]_B,
    \end{array}\right.\] we obtain the equivalence between the following
assertions:

\begin{enumerate}
\def\labelenumi{(\roman{enumi})}
\item
  \((H, \Lambda)\) solves
  \(\operatorname{ICS} (X, \operatorname{Cov}_{w_1}, \operatorname{Cov}_{w_2})\)
  in the space \(E\)
\item
  \((G_B^{1/2} [H]_B, \Lambda)\) solves
  \(\operatorname{ICS} (G_B^{1/2} [X]_B, \operatorname{Cov}_{w_1}, \operatorname{Cov}_{w_2})\)
  in the space \(\mathbb R^p\),
\end{enumerate}

which gives the equivalence between the assertions (1) and (2). The
equivalence between the other assertions are deduced from the fact that
for any \(\ell \in \{ 1, 2 \}\) and any \((x,y) \in E^2\):
\begin{equation}\phantomsection\label{eq-cor-coord}{\begin{aligned}
\langle \operatorname{Cov}_{w_\ell}^E [X] x, y \rangle_E
&= \langle \operatorname{Cov}_{w_\ell} (G_B^{1/2} [X]_B) G_B^{1/2} [x]_B, G_B^{1/2} [y]_B \rangle_{\mathbb R^p} \\
&= \langle \operatorname{Cov}_{w_\ell} (G_B [X]_B) [x]_B, [y]_B \rangle_{\mathbb R^p} \\
&= \langle \operatorname{Cov}_{w_\ell} ([X]_B) G_B [x]_B, G_B [y]_B \rangle_{\mathbb R^p},
\end{aligned}}\end{equation} where Equation~\ref{eq-cor-coord} (1) comes
from the equation Equation~\ref{eq-isometry-covw}, and the equalities
Equation~\ref{eq-cor-coord} (2) and Equation~\ref{eq-cor-coord} (3) come
from the affine equivariance of \(\operatorname{Cov}_{w_\ell}\).~
\end{proof}

\begin{proof}[Proposition~\ref{prp-reconstruction}]
Let us decompose \(S_1[X]^{-1} (X - \mathbb EX)\) over the basis \(H\),
which is orthonormal in \((E, \langle \cdot, S_1[X] \cdot \rangle)\):
\[\begin{aligned}
S_1[X]^{-1} (X - \mathbb EX)
&= \sum_{j=1}^p \langle S_1[X]^{-1} (X - \mathbb EX), S_1[X] h_j \rangle h_j \\
&= \sum_{j=1}^p \langle X - \mathbb EX, h_j \rangle h_j \\
S_1[X]^{-1} (X - \mathbb EX)
&= \sum_{j=1}^p z_j h_j.
\end{aligned}\] The dual basis \(H^*\) of \(H\) is the one that
satisfies \(\langle h_j, h^*_{j'} \rangle = \delta_{jj'}\) for all
\(1 \leq j, j' \leq p\) and we know from the definition of ICS that this
holds for \((S_1[X] h_j)_{1 \leq j \leq p}\).~
\end{proof}

\subsection*{Code \& reproducibility}\label{code-reproducibility}
\addcontentsline{toc}{subsection}{Code \& reproducibility}

In order to implement coordinate-free ICS, we created the R package
\texttt{ICSFun}, which is used to generate the figures (see the code in
this HTML version of the article).

\printbibliography

\end{document}